\documentclass{scrartcl}

\usepackage[utf8]{inputenc}

\usepackage{amsmath}
\usepackage{amssymb}
\usepackage{bm}
\usepackage{mathrsfs}
\usepackage{csquotes}
\usepackage{subcaption}

\usepackage{tikz}
\usetikzlibrary{calc}
\usetikzlibrary{shapes}
\usetikzlibrary{automata, positioning, decorations.pathmorphing}

\usepackage{amsthm}
\usepackage{thmtools}

\theoremstyle{plain}
\newtheorem{theorem}{Theorem}[section]

\newtheorem{proposition}[theorem]{Proposition}
\newtheorem{corollary}[theorem]{Corollary}
\newtheorem{fact}[theorem]{Fact}

\theoremstyle{definition}
\newtheorem{definition}[theorem]{Definition}
\newtheorem{example}[theorem]{Example}

\theoremstyle{remark}
\newtheorem{remark}[theorem]{Remark}
\newtheorem*{remark*}{Remark}

\usepackage{calc}
\newlength{\edgelength}
\newcommand{\trans}[4]{%
  \begin{tikzpicture}[auto, shorten >=1pt, >=latex, baseline=(l.base), inner sep=0pt, outer xsep=0.3333em]
    \node (l) {\ensuremath{#1}};%
    \setlength{\edgelength}{\widthof{\scriptsize\ensuremath{#2/#3}}+0.5cm}%
    \node[base right=\edgelength of l] (r) {\ensuremath{#4}};%
    \path[->] (l.mid east) edge node[inner sep=0pt] {\scriptsize\ensuremath{#2/#3}} (r.mid west);%
  \end{tikzpicture}%
}

\newcommand{\transa}[3]{%
  \begin{tikzpicture}[auto, shorten >=1pt, >=latex, baseline=(l.base), inner sep=0pt, outer xsep=0.3333em]
    \node (l) {\ensuremath{#1}};%
    \setlength{\edgelength}{\widthof{\scriptsize\ensuremath{#2}}+0.5cm}%
    \node[base right=\edgelength of l] (r) {\ensuremath{#3}};%
    \path[->] (l.mid east) edge node[inner xsep=0pt, inner ysep=0.2em] {\scriptsize\ensuremath{#2}} (r.mid west);%
  \end{tikzpicture}%
}

\newcommand{\ltrans}[4]{%
  \begin{tikzpicture}[auto, shorten >=1pt, >=latex, baseline=(l.base), inner sep=0pt, outer xsep=0.3333em]
    \node (l) {\ensuremath{#1}};%
    \setlength{\edgelength}{\widthof{\scriptsize\ensuremath{#2/#3}}+0.5cm}%
    \node[base right=\edgelength of l] (r) {\ensuremath{#4}};%
    \path[<-] (l.mid east) edge node[inner sep=0pt] {\scriptsize\ensuremath{#2/#3}} (r.mid west);%
  \end{tikzpicture}%
}

\newcommand{\ltransa}[3]{%
  \begin{tikzpicture}[auto, shorten >=1pt, >=latex, baseline=(l.base), inner sep=0pt, outer xsep=0.3333em]
    \node (l) {\ensuremath{#1}};%
    \setlength{\edgelength}{\widthof{\scriptsize\ensuremath{#2}}+0.5cm}%
    \node[base right=\edgelength of l] (r) {\ensuremath{#3}};%
    \path[<-] (l.mid east) edge node[inner xsep=0pt, inner ysep=0.2em] {\scriptsize\ensuremath{#2}} (r.mid west);%
  \end{tikzpicture}%
}

\tikzset{path/.style = {decorate, decoration={snake, pre length=1mm, post length=2mm, amplitude=0.3mm, segment length=1mm}}}

\newcommand{\run}[4]{%
  \begin{tikzpicture}[auto, shorten >=1pt, >=latex, baseline=(l.base), inner sep=0pt, outer xsep=0.3333em]
    \node (l) {\ensuremath{#1}};%
    \setlength{\edgelength}{\widthof{\scriptsize\ensuremath{#2/#3}}+0.5cm}%
    \node[base right=\edgelength of l] (r) {\ensuremath{#4}};%
    \path[->] ([yshift=-0.25mm] l.mid east) edge[path] node[above, inner sep=0pt] {\scriptsize\ensuremath{#2/#3}} ([yshift=-0.25mm] r.mid west);%
  \end{tikzpicture}%
}

\newcommand{\lrun}[4]{%
  \begin{tikzpicture}[auto, shorten >=1pt, >=latex, baseline=(l.base), inner sep=0pt, outer xsep=0.3333em]
    \node (l) {\ensuremath{#1}};%
    \setlength{\edgelength}{\widthof{\scriptsize\ensuremath{#2/#3}}+0.5cm}%
    \node[base right=\edgelength of l] (r) {\ensuremath{#4}};%
    \path[->] ([yshift=-0.25mm] r.mid west) edge[path] node[above, inner sep=0pt] {\scriptsize\ensuremath{#2/#3}} ([yshift=-0.25mm] l.mid east);%
  \end{tikzpicture}%
}

\newcommand{\lruna}[3]{%
  \begin{tikzpicture}[auto, shorten >=1pt, >=latex, baseline=(l.base), inner sep=0pt, outer xsep=0.3333em]
    \node (l) {\ensuremath{#1}};%
    \setlength{\edgelength}{\widthof{\scriptsize\ensuremath{#2}}+0.5cm}%
    \node[base right=\edgelength of l] (r) {\ensuremath{#3}};%
    \path[->] ([yshift=-0.25mm] r.mid west) edge[path] node[above, inner sep=0pt] {\scriptsize\ensuremath{#2}} ([yshift=-0.25mm] l.mid east);%
  \end{tikzpicture}%
}
\usepackage{xspace}

\newcommand*{\SAut}{$\mathscr{S}$\kern-0.4ex-automaton\xspace}
\newcommand*{\SAuta}{$\mathscr{S}$\kern-0.4ex-automata\xspace}
\newcommand*{\SIAut}{$\inverse{\mathscr{S}}$\kern-0.4ex-automaton\xspace}
\newcommand*{\SIAuta}{$\inverse{\mathscr{S}}$\kern-0.4ex-automata\xspace}
\newcommand*{\GAut}{$\mathscr{G}$\kern-0.2ex-automaton\xspace}
\newcommand*{\GAuta}{$\mathscr{G}$\kern-0.2ex-automata\xspace}
\usepackage{tabularx}
\newcommand{\problem}[3][]{%
  \par\vspace{0.125cm plus 0.1cm minus 0.05cm}\begin{tabularx}{\textwidth-2\parindent}{lX}%
    \if\relax\detokenize{#1}\relax%
    \else%
      \textnormal{\textbf{Constant:}}&#1\\%
    \fi%
    \textnormal{\textbf{Input:}}&#2\\%
    \textnormal{\textbf{Question:}}&#3\\%
  \end{tabularx}\vspace{0.125cm plus 0.1cm minus 0.05cm}\par%
  }
  
\newcommand*{\ComplexityClass}[1]{\textsc{#1}}

\usepackage{hyperref}

\usepackage[affil-it]{authblk}

\author{Daniele D'Angeli\thanks{The first two authors are members of the INdAM-GNSAGA group.}}
\affil{Università degli Studi Niccolò Cusano\\
  Via Don Carlo Gnocchi, 3\\
  00166 Roma, Italy}
\author{Emanuele Rodaro\textsuperscript{$*$}}
\affil{Department of Mathematics\\
  Politecnico di Milano\\
  Piazza Leonardo da Vinci, 32\\
  20133 Milano, Italy}
\author{Jan Philipp Wächter\thanks{The third author was funded by the Deutsche Forschungsgemeinschaft (DFG, German Research Foundation) – 492814705 – while visiting the Department of Mathematics at Politecnico di Milano. Some part of this article was written while he was affiliated with Universität des Saarlandes and partly funded by ERC grant 101097307. The listed affiliation is his current one where his research is supported by EPSRC project EP/Y008626/1 'Special Inverse Monoids: Geometry, Structure \& Algorithms'.}}
\affil{Department of Mathematics\\
  University of Manchester\\
  Oxford Road\\
  Manchester M13 9PL, UK}

\title{The Finiteness Problem for Automaton Semigroups of Extended Bounded Activity}

\begin{document}
  \maketitle
  \begin{abstract}
    We extend the notion of activity for automaton semigroups and monoids introduced by Bartholdi, Godin, Klimann and Picantin to a more general setting. Their activity notion was already a generalization of Sidki's activity hierarchy for automaton groups.
    We show that the language of $\omega$-words with infinite orbits is effectively a deterministic Büchi language for automata with bounded extended activity,
    which yields decidability of the finiteness problem for complete automaton semigroups and monoids of bounded activity (solving an open problem by Bartholdi, Godin, Klimann and Picantin).
    In fact, we obtain a stronger result also covering finitely generated subsemigroups.
    
  \end{abstract}
  
  \begin{section}{Introduction}
    The class of automaton groups is a rich source for groups with exotic properties. Probably the most famous one of these groups is Grigorchuk's group (and we refer the reader e.\,g.\ to \cite{bartholdi2021groups,nekrashevych2005self} for many more examples). Grigorchuk's group was the first example of a group with intermediate growth (see e.\,g.\ \cite{grigorchuk2008groups} for more information) and is also a Burnside group (in the sense that it is infinite while every element has finite order) as well as amenable but not elementary amenable (see e.\,g.\ \cite{juschenko2022amenability}). Interestingly, it is not finitely presented \cite{grigorchuk1999system} but, as an automaton group, it has a finite description using an automaton. This demonstrates that using automata to generate groups (and also semigroups) is algorithmically interesting as we may give the generating automaton as a finite input to an algorithm, which allows us to consider algorithms for groups that do not have finite presentations (in the sense of a pair of finitely many generators and finitely many relations over them).
    Semigroups arise naturally in this field (e.\,g.\ via the dual automaton) and have received quite a lot of attention structurally (e.\,g.\ \cite{cain2009automaton, brough2015automaton, klimann2016automaton, brough2017automaton, picantin2019automatic, bartholdi2020hierarchy, structurePart, orbitsPart, decidabilityPart, brough2023automaton}) and algorithmically (e.\,g.\ \cite{klimann2012implementing, gillibert2014finiteness, bartholdi2020hierarchy, expandabilityPart, dangeli2024freeness}).
    
    An \emph{automaton} in this context is a deterministic and complete finite-state, letter-to-letter transducer (without initial or final states), i.\,e.\ a finite graph whose edges are labeled by pairs $(a, b)$ of an input letter $a$ and an output letter $b$. The idea is that we may start reading an input word $u$ in some state $p$ (i.\,e.\ some node) and obtain a corresponding output word $v$. This way, we may associate to $p$ a function that maps $u$ to $v$ and consider the closure under composition of these functions for all states as the semigroup generated by the automaton.
    For invertible automata, these functions are bijections and we naturally obtain a group.
    Groups and semigroups arising in this way are called \emph{automaton groups} and \emph{automaton semigroups}.
    
    More structure was brought into this by Sidki's \emph{activity} notion \cite{sidki2000automorphisms} for automaton groups. His idea was to look at the structure of the cycles in the automaton (see \cite{sidki2000automorphisms} or e.\,g.\ \cite{waechter2024word} for more details). A word $u$ here is said to be \emph{active} for a state $p$ if, after reading $u$ starting in $p$, we do not end in a/the state acting as the identity on all words. If the number of active words of any length is bounded uniformly by a constant, we say that $p$ has \emph{bounded} activity. If the number of words grows linearly with their length, we speak of \emph{linear} activity. If it grows quadratically (cubically, etc.), also the activity is \emph{quadratic} (\emph{cubic}, etc.). It turns out, that the activity is either polynomial (with a natural number as the degree) or exponential.
    
    The activity of a single state naturally generalizes to the activity of a group element (which is typically given as a word over the states as generators). If two group elements have, say, bounded activity, also their product has, which allows us to speak, for example, of \emph{bounded automaton groups} (in which all elements have bounded activity). This yields an infinite hierarchy within the class of automaton groups, which consists of a polynomial and an exponential part. Interestingly, many famous and interesting examples of automaton groups are already contained in the second-lowest level of bounded activity (the lowest level belonging to \emph{finitary} automata consists precisely of the class of finite groups in the sense that every finite automaton generates a finite group and every finite group is generated by a finitary automaton; compare to \autoref{ex:finiteMonoid}). This applies, in particular, to Grigorchuk's group (and further examples, see e.\,g.\ in \cite{waechter2024word}).
    
    Maybe surprisingly, bounded automaton groups (despite their already quite complicated nature) seem to be \enquote{finite enough} such that many of their decision problems remain decidable (or have lower complexity compared to general automaton groups). For example, the order problem is decidable in such groups \cite{bondarenko2013conjugacy} while it is undecidable for general automaton groups \cite{gillibert2018automaton,bartholdi2017word}. Their finiteness problem is also decidable \cite{bondarenko21orbits} while the problem is widely suspected to be undecidable in the general case but sill open \cite[7.2.(b)]{grigorchuk2000automata} (it is known to be undecidable for automaton semigroups \cite{gillibert2014finiteness}). Furthermore, bounded automaton groups are contracting \cite{bondarenko2003postCritically} (see also e.\,g.\ \cite{nekrashevych2005self,waechter2024word}) and their word problem is therefore solvable in logarithmic space (by a deterministic Turing machine) while, for general automaton groups, there is one with $\ComplexityClass{PSpace}$-complete word problem \cite{waechter2023automaton} (see \cite{waechter2024word} for more background information; see also \cite{kotowsky2023word} for the word problem on the lowest level of the hierarchy).
    
    Likely driven by such algorithmic results (as they show that the order/torsion problem remains decidable), Bartholdi, Godin, Klimann and Picantin \cite{bartholdi2020hierarchy} generalized the notion of activity to automaton monoids.\footnote{Although they do not clearly distinguish between automaton monoids and semigroups.} In order to keep certain desirable properties of the activity notion in the group case, their generalization rather uses the output instead of the input for defining active words, which makes the geometric characterization based on the cyclic structure of the generating automaton less transparent. However, they still count those words that do not end in the/an identity state, which obviously is only a useful concept if such an identity state exists; in this case, however, the generated semigroup is necessarily a monoid. As a potential extension (that is also interesting for semigroups), they propose to replace the identity state by what is called the maximal NoCyWEx subautomaton (whose definition is a bit technical but which necessarily generates a finite subsemigroup).
    
    However, there does not seem to be any reason to only consider this particular subautomaton. Instead, we introduce a different notion of extended activity where we may consider any subset $S$ of states (which is not a restriction as any finite set of semigroup elements may always be considered to appear as single states) that is closed in the sense that we may not leave $S$ by reading any word in the generators. A(n output) word is then \emph{$S$-active} if we do not end in a state in $S$ after reading it. This strictly generalizes all previous activity notions: if we let $S$ contain only the identity state, we obtain the activity notion for groups and monoids and, if we take $S$ as the NoCyWEx part of the automaton, we obtain the above extended activity.
    
    We then consider automata which have bounded $S$-activity. We show that, if we are given the promise that $S$ generates a finite subsemigroup (as the finiteness problem for automaton semigroups is undecidable \cite{gillibert2014finiteness}), then the finiteness problem for automaton semigroups of bounded $S$-activity is decidable. This works, in particular, for all the other activity notions and generalizes the corresponding group result \cite{bondarenko21orbits} to semigroups and monoids (solving an open problem stated by Bartholdi, Godin, Klimann and Picantin \cite{bartholdi2020hierarchy}). In fact, we obtain this algorithmic result by showing that the language of $\omega$-words with an infinite orbit is recognized (by an effectively constructible) deterministic Büchi acceptor (which also is a generalization from the group case in \cite{bondarenko21orbits}). The construction of this Büchi acceptor is based on the notion of \emph{expandability} introduced by the authors in \cite{expandabilityPart} and the final connection to the finiteness problem is due to the fact that an automaton semigroup is infinite if and only if it admits an (infinite) word with an infinite orbit. As a by-product of our characterization of words with infinite orbits as Büchi languages, we obtain that an automaton semigroup of bounded $S$-activity is infinite if and only if it admits an ultimately periodic word with an infinite orbit.
    
    In fact, we obtain these results not only for the full orbits but also for the case where we consider the sub-orbits given by a regular, suffix-closed language of the state set. This shows, in particular, that the finiteness problem for finitely generated subsemigroups of automaton semigroups with bounded $S$-activity is decidable (which the order problem is a special case of; compare to \cite{bartholdi2020hierarchy}).
    
    We are confident that the techniques we use to obtain our results can also be refined to cover further algorithmic (and non-algorithmic) problems.
  \end{section}
  
  \begin{section}{Preliminaries}
    \paragraph{Sets, Semigroups, Words and (Regular) Languages.}
    We write $\mathcal{P}(A)$ for the power set of a set $A$ and $A \uplus B$ for the disjoint union of two sets $A$ and $B$. A set $S$ with an associative operation is a \emph{semigroup} and, if $S$ additionally has a neutral element (i.\,e.\ an element $e \in S$ with $es = s = se$ for all $s \in S$), it is a \emph{monoid}. We assume the reader to be familiar with the basics of semigroup theory but we will not need any advanced concepts.
    
    An \emph{alphabet} is a non-empty, finite set $\Sigma$. By $\Sigma^*$, we denote the set of \emph{finite words} over $\Sigma$ including the empty word $\varepsilon$. If we want to exclude it, we write $\Sigma^+$.
    For the set of finite words of length exactly $n \geq 0$, we write $\Sigma^n$.
    By $\Sigma^\omega$, we denote the set of \emph{$\omega$-words} (i.\,e.\ \emph{right infinite words}) over $\Sigma$. We use the terms \emph{word} for both, finite words and $\omega$-words and a \emph{language} can either be a subset of $\Sigma^*$ or a subset of $\Sigma^\omega$.
    
    \paragraph*{Regular Languages.}
    For a language $L \subseteq \Sigma^*$ (of finite words), we define the (Myhill-Nerode) relation $u \mathrel{L} v \iff (\forall x \in \Sigma^*: ux \in L \iff vx \in L)$. The classes of this relation are called the (Myhill-Nerode) \emph{classes} of $L$ and the class of $u \in \Sigma^*$ is denoted by $\mathscr{C}_L(u)$ or simply $\mathscr{C}(u)$ if the language is clear from the context.\footnote{The classes of $L$ can be identified with the states of the minimal deterministic acceptor of $L$. The initial state is $\mathscr{C}(\varepsilon)$ and the final states are the classes $C$ with $C \subseteq L$. For every class $\mathscr{C}(u)$ and every $a \in \Sigma$, we have an $a$-labeled transition from $\mathscr{C}(u)$ to $\mathscr{C}(ua)$.} A language of finite words is \emph{regular} if it has finitely many classes.
    
    \paragraph*{Büchi Acceptors and $\omega$-Regular Languages.}
    An (edge-accepting) \emph{Büchi acceptor} (BA) is a tuple $\mathcal{B} = (Z, \Gamma, \tau, z_0, \mathcal{F})$ where $Z$ is a finite set of \emph{states}, $\Gamma$ is an alphabet, $\tau \subseteq Z \times \Gamma \times Z$ is a set of \emph{transitions}, $z_0 \in Z$ is the \emph{initial} state and $\mathcal{F} \subseteq \tau$ is the set of \emph{accepting} transitions.
    In the context of transitions, we also write $\transa{y}{a}{z}$ for the element $(y, a, z) \in Z \times \Gamma \times Z$.
    
    A Büchi acceptor $\mathcal{B} = (Z, \Gamma, \tau, z_0, \mathcal{F})$ is \emph{deterministic} if we have
    \[
      d_{y, a} = \left| \{ \transa{y}{a}{z} \in \tau \mid z \in Z \} \right| = 1
    \]
    for all $y \in Z$ and $a \in \Gamma$.
    
    A \emph{run} of the Büchi acceptor $\mathcal{B}$ is an infinite sequence
    \begin{center}
      \begin{tikzpicture}[auto, shorten >=1pt, >=latex, baseline=(p0.base), inner sep=0pt, outer xsep=0.3333em]
        \setlength{\edgelength}{\widthof{\scriptsize\ensuremath{a_n}}+0.5cm}%
        
        \node (p0) {$y_0$};%
        \node[base right=\edgelength of p0] (p1) {$y_1$};
        \node[base right=\edgelength of p1] (dots) {$\dots$};
        
        \path[->]
          (p0.mid east) edge node[inner xsep=0pt, inner ysep=0.2em] {\scriptsize$a_1$} (p1.mid west)
          (p1.mid east) edge node[inner xsep=0pt, inner ysep=0.2em] {\scriptsize$a_2$} (dots.mid west)
        ;
      \end{tikzpicture}
    \end{center}
    where $\transa{y_{i - 1}}{a_i}{y_i} \in \tau$ for all $i \geq 1$.
    The input of this run is $a_1 a_2 \dots \in \Gamma^\omega$.
    Such a run is \emph{initial} if $y_0 = z_0$ and it is \emph{accepting} if $\transa{y_{i - 1}}{a_i}{y_i} \in \mathcal{F}$ for infinitely many $i \geq 1$.
    The Büchi acceptor \emph{accepts} an $\omega$-word $\alpha \in \Gamma^\omega$ if it admits an accepting initial run whose input is $\alpha$.
    The language (of $\omega$-words) of accepted words is the language \emph{recognized} (or \emph{accepted}) by the Büchi acceptor.
    A language is \emph{$\omega$-regular} if it is accepted by some Büchi acceptor and, if the Büchi acceptor is deterministic, the language is \emph{deterministic} $\omega$-regular.
    \begin{remark}
      Instead of defining the acceptance using accepting transitions we could have used the more common acceptance criterion based on accepting states (where an $\omega$-word $\alpha$ is accepted if there is an initial run of the above form with input $\alpha$ such that infinitely many of the visited states $y_i$ are accepting).
      It is not difficult to see, though, that the two notions are equivalent: A state-accepting BA is tuned into an edge-accepting one by marking all transitions going into an accepting state as accepting. In the other direction, we duplicate every state and have an accepting and a non-accepting version. Now, we let accepting transitions end in the accepting one and non-accepting transitions end in the corresponding non-accepting state.
    \end{remark}
    \begin{remark}
      It is well-known that the class of deterministic $\omega$-regular languages is a proper subclass of all $\omega$-regular languages. See, for example, \cite{perrin2004infinite} for more information on $\omega$-regular languages.
    \end{remark}
    
    \paragraph*{Semigroup Presentations and Free Products.}
    A \emph{semigroup presentation} is a pair $\langle Q \mid \mathcal{R} \rangle_\mathscr{S}$ of a set of \emph{generators} $Q$ and a (possibly infinite) set of \emph{relations} $\mathcal{R} \subseteq Q^+ \times Q^+$. Typically, we will only consider finitely generated semigroups (i.\,e.\ we assume $Q$ to be a finite set and, usually, also non-empty).
    If we denote by $\mathcal{C}$ the smallest congruence $\mathcal{C} \subseteq Q^+ \times Q^+$ with $\mathcal{R} \subseteq \mathcal{C}$, the semigroup \emph{presented} by such a presentation is $S = Q^+ / \mathcal{C}$ formed by the congruence classes $[\cdot]_{\mathcal{C}}$ of $\mathcal{C}$ with the (well-defined!) operation $[ u ]_{\mathcal{C}} \cdot [v]_{\mathcal{C}} = [uv]_{\mathcal{C}}$. Every semigroup generated by a finite, non-empty set $Q$ is presented by some semigroup presentation of this form.
    
    The free product of the semigroups $S = \langle Q \mid \mathcal{S} \rangle_\mathscr{S}$ and $T = \langle P \mid \mathcal{R} \rangle_\mathscr{S}$ is the semigroup $S \star T = \langle Q \uplus P \mid \mathcal{S} \cup \mathcal{R} \rangle_\mathscr{S}$. For example, we have $\{ p, q \}^+ = p^+ \star q^+$. Any element of $S \star T$ can be written as a non-empty sequence of blocks $\bm{p}_0 \bm{q}_1 \bm{p}_1 \ldots \bm{q}_k \bm{p}_k$ where $\bm{p}_0, \bm{p}_k \in P^*$ may be empty but the other blocks $\bm{q}_i \in Q^+$ (for $1 \leq i \leq k$) and $\bm{p}_i \in P^+$ (for $1 \leq i < k$) may not.
    \begin{remark*}
      Of course, there is also the free product of monoids (and monoid presentations). However, in this paper, we only consider free products of semigroups (in particular, we have $\{ p, q \}^* \not\simeq p^* \star q^*$).
    \end{remark*}
    
    \paragraph{\SAuta.}
    In the setting of the current paper, an \emph{automaton}\footnote{In more general automaton-theoretic terms, this would rather be called a finite-sate, letter-to-letter transducer.} is a tuple $\mathcal{T} = (Q, \Sigma, \delta)$ where $Q$ is the finite, non-empty set of \emph{states}, $\Sigma$ is an alphabet and $\delta \in Q \times \Sigma \times \Sigma \times Q$ is a set of \emph{transitions}. In the context of transitions, we also use the graphical notations $\trans{p}{a}{b}{q}$ and $\ltrans{q}{a}{b}{p}$ for the tuple $(p, a, b, q) \in Q \times \Sigma \times \Sigma \times Q$. The intuitive idea is that, if this element is contained in $\delta$, we may read the \emph{input} letter $a$ \emph{starting} in $p$ and obtain the \emph{output} letter $b$ and \emph{end} in $q$.
    
    A \emph{run} in an automaton $\mathcal{T} = (Q, \Sigma, \delta)$ is a sequence
    \begin{center}
      \begin{tikzpicture}[auto, shorten >=1pt, >=latex, baseline=(l.base), inner sep=0pt, outer xsep=0.3333em]
        \node (p0) {\ensuremath{p_0}};%
        \setlength{\edgelength}{\widthof{\scriptsize\ensuremath{a_\ell/b_\ell}}+0.5cm}%
        \node[base right=\edgelength of p0] (d) {\ensuremath{\dots}};%
        \path[->] (p0.mid east) edge node[inner sep=0pt] {\scriptsize\ensuremath{a_1/b_1}} (d.mid west);%
        \node[base right=\edgelength of d] (pl) {\ensuremath{p_\ell}};%
        \path[->] (d.mid east) edge node[inner sep=0pt] {\scriptsize\ensuremath{a_\ell/b_\ell}} (pl.mid west);%
      \end{tikzpicture}
    \end{center}
    with $\trans{p_{i - 1}}{a_i}{b_i}{p_i} \in \delta$ for all $0 < i \leq \ell$ (where $\ell \geq 0$).
    It \emph{starts} in $p_0$, \emph{ends} in $p_\ell$, its input is $a_1 \dots a_\ell$ and its output is $b_1 \dots b_\ell$.
    We will abbreviate such a run also as $\run{p_0}{a_1 \dots a_\ell}{b_1 \dots b_\ell}{p_\ell}$.
    
    An automaton $\mathcal{T} = (Q, \Sigma, \delta)$ is a \emph{complete \SAut} if we have
    \[
      d_{p, a} = \left| \{ \trans{p}{a}{b}{q} \in \delta \mid b \in \Sigma, q \in Q \} \right| = 1
    \]
    for all $p \in Q$ and $a \in \Sigma$.\footnote{i.\,e.\ if it is deterministic and complete} It is easy to see that, in such an automaton, there is a unique run starting in $p$ with input $u$ for every $p \in Q$ and $u \in \Sigma^*$. This allows us to uniquely define $p \circ u$ as the output and $p \cdot u$ as the end of this run (which is often called the \emph{section} or \emph{restriction} of $p$ at $u$ in the literature). This can naturally be extended into a left action of $Q^*$ on $\Sigma^*$ by letting $\varepsilon \circ u = u$ and, inductively, $\bm{q} p \circ u = \bm{q} \circ (p \circ u)$ for $p \in Q$ and $\bm{q} \in Q^*$. Clearly, we have $q_\ell \dots q_1 \circ u = q_\ell \circ \dots \circ q_1 \circ u$.
    
    Similarly, we may also extend the notation $p \cdot u$ into a right action of $\Sigma^*$ on $Q^*$, which is called the \emph{dual action}, by letting $\varepsilon \cdot u = \varepsilon$ and
    \[
      \bm{q} p \cdot u = \left[ \bm{q} \cdot (p \circ u) \right] \, (p \cdot u) \text{.}
    \]
    
    For any complete \SAut $\mathcal{T} = (Q, \Sigma, \delta)$, we define the relation ${=_\mathcal{T}} \subseteq Q^* \times Q^*$ by 
    \[
      \bm{p} =_\mathcal{T} \bm{q} \iff \forall w \in \Sigma^*: \bm{p} \circ w = \bm{q} \circ w \textbf{.}
    \]
    It is easy to see that ${=_\mathcal{T}}$ is a congruence. Thus, $Q^+/{{=_\mathcal{T}}}$ is a semigroup $\mathscr{S}(\mathcal{T})$ and we say that it is the \emph{semigroup generated by $\mathcal{T}$}. Any semigroup arising in this way is called an \emph{automaton semigroup}.
    
    In addition to ${=_\mathcal{T}}$, we will use similar notation. In particular, we will write $\bm{q} \in_\mathcal{T} P$ for $\bm{q} \in Q^*$ and $P \subseteq Q^*$ if there is some $\bm{p} \in P$ with $\bm{q} =_\mathcal{T} \bm{p}$. For $P \subseteq Q^+$ and $\bm{q} \in Q^+$, writing $\bm{q} \in_\mathcal{T} P^+$ means that the image of $\bm{q}$ in the semigroup $\mathscr{S}(\mathcal{T})$ is contained in the subsemigroup generated by $P$.
    
    \begin{example}[The Adding Machine]\label{ex:addingMachine}
      The classical example of a complete \SAut is the \emph{adding machine}:
      
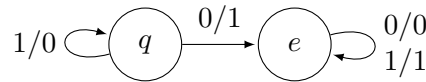
\begin{center}
        \begin{tikzpicture}[auto, shorten >=1pt, >=latex]
          \node[state] (q) {$q$};
          \node[state, right=of q] (id) {$e$};
          
          \draw[->] (q) edge[loop left] node {$1/0$} (q)
                        edge node {$0/1$} (id)
                    (id) edge[loop right] node[align=center] {$0/0$\\$1/1$} (id);
        \end{tikzpicture}
      \end{center}
      Clearly, we have $e \circ u = u$ for all $u \in \{ 0, 1 \}^*$. In order to understand the action of $q$, it is best to look at an example. We have:
      \begin{align*}
        q \circ 000 &= 100
        & q^3 \circ 000 &= q \circ 010 = 110 \\
        q^2 \circ 000 &= q \circ 100 = 010
        & q^4 \circ 000 &= q \circ 110 = 001
      \end{align*}
      More generally, if we denote the binary representation of $n \in \mathbb{N}$ of length $\ell$ in reverse/with the least significant bit on the left as $\operatorname{bin}_\ell n$, we have $q \circ \operatorname{bin}_\ell n = \operatorname{bin}_\ell (n + 1)$ for all $n \in \mathbb{N}$ (with sufficiently large $\ell$). This shows ${q}^i \neq_{\mathcal{T}} q^j$ for all $i \neq j$ (with $i, j > 0$). By identifying $q^0$ with $e$, we obtain that the semigroup generated by the adding machine is (isomorphic to) the free monogenic monoid $q^*$.
    \end{example}
    
    \begin{example}[Finite Monoids]\label{ex:finiteMonoid}
      Consider the finite monoid $M = \{ e, p, q \}$ with $e^2 = e$, $ep = p = pe$, $eq = q = qe$, $pq = p = pp$ and $qp = q = qq$ in $M$.
      In particular, we have $ps = p$ and $qs = q$ for all $s \in M$.
      Consider the complete \SAut\footnote{We use this notation to mean that all transitions starting in $p$ have output $p$ and all transitions starting in $q$ have output $q$.} $\mathcal{T} = (M, M, \delta)$
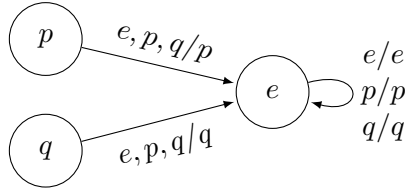
      \begin{center}
        \begin{tikzpicture}[auto, shorten >=1pt, >=latex]
          \node[state] (p) {$p$};
          \node[state, below=0.5cm of p] (q) {$q$};
          \node[state, anchor=base] (e) at ($(p.base)!0.5!(q.base)+(3cm, 0pt)$) {$e$};
          
          \draw[->] (p) edge[sloped] node[align=center] {$e, p, q / p$} (e)
                    (q) edge[sloped] node[swap, align=center] {$e, p, q / q$} (e)
                    (e) edge[loop right] node[align=center] {$e/e$\\$p/p$\\$q/q$} (e)
          ;
        \end{tikzpicture}
      \end{center}
      where $M$ is both the state set and the alphabet.
      For a word $w \in M^* \cup M^\omega$ and a letter $a \in M$, we have
      \[
        e \circ aw = aw, \quad p \circ aw = pw \quad \text{and} \quad q \circ aw = qw \text{.}
      \]
      This shows that we have all the relations of $M$ also in $\mathscr{S}(\mathcal{T})$: $e$ is the neutral element and, for any $s \in M$, we have $ps \circ aw = p \circ (s \circ a)w = pw = p \circ aw$ and, thus, $ps =_\mathcal{T} p$ (and an analogous statement for $q$).
      
      This construction generalizes to any finite monoid $M$ (and, thus, also for any finite group):
      It is generated by the complete \SAut $\mathcal{T} = (M, M, \delta)$ with
      \[
        \delta = \{ \trans{m}{n}{mn}{e} \mid m, n \in M \} \text{.}
      \]
    \end{example}
    
    \begin{example}[Finite Semigroups]\label{ex:finiteSemigroup}
      The previous example heavily used the neutral element of the monoid.
      However, any finite semigroup $S$ (even if it is not a monoid) is generated by a complete \SAut (this is due to \cite[Proposition~4.6]{cain2009automaton}).
      The construction for this first adjoins a new element $e$ to $S$ with $e^2 = e$ and $es = s = se$ for all $s \in S$ to obtain the monoid $S^e$.
      It is easy to see that $S$ acts faithfully on $S^e$ by left translation (see e.\,g.\ \cite[Theorem 1.1.2]{howie1995fundamentals}), which basically shows that $(S, S^e, \delta)$ with
      \begin{align*}
        \delta &= \{ \trans{s}{t}{st}{s} \mid s \in S, t \in S^e \}
      \end{align*}
      is a complete \SAut generating $S$ (see \cite[Proposition~4.6]{cain2009automaton} for a full proof).
    \end{example}
    
    \begin{example}[Grigorchuk's Group]\label{ex:grigorchuk}
      A very famous automaton is the complete \SAut
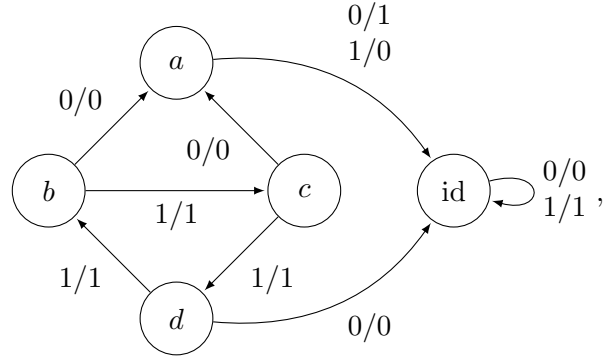
      \begin{center}
        \begin{tikzpicture}[auto, shorten >=1pt, >=latex, baseline=(id.base)]
          \node[state] (b) {$b$};
          \node[state, above right=of b] (a) {$a$};
          \node[state, below right=of b] (d) {$d$};
          \node[state, below right=of a] (c) {$c$};
          \node[state, right=of c] (id) {$\textnormal{id}$};
          
          \draw[->] (a) edge[bend left] node[align=center] {$0/1$\\$1/0$} (id)
                    (b) edge node {$0/0$} (a)
                    (b) edge node[swap] {$1/1$} (c)
                    (c) edge node {$0/0$} (a)
                    (c) edge node {$1/1$} (d)
                    (d) edge[bend right] node[swap] {$0/0$} (id)
                    (d) edge node {$1/1$} (b)
                    (id) edge[loop right] node[align=center] {$0/0$\\$1/1$} (id)
                    ;
        \end{tikzpicture},
      \end{center}
      whose generated semigroup is Grigorchuk's group (which coincides also with the generated monoid and group).
      We will not elaborate on this further but refer the reader, for example to
      \cite{grigorchuk2008groups, nekrashevych2005self} for details.
    \end{example}
    
    \begin{example}[Running Example]\label{ex:runningAut}
      We will use the following complete \SAut as a running example for demonstrating our constructions:
      
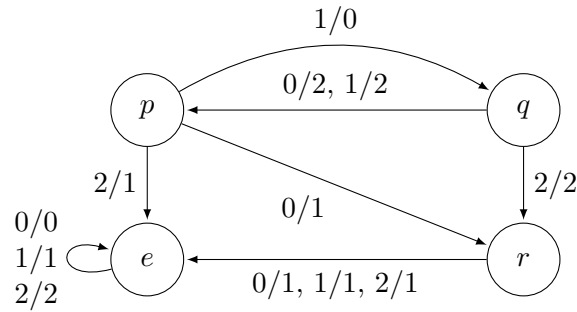
\begin{center}
        \begin{tikzpicture}[auto, shorten >=1pt, >=latex]
          \node[state] (p) {$p$};
          \node[state, right=4cm of p] (q) {$q$};
          \node[state, below=of p] (e) {$e$};
          \node[state, below=of q] (r) {$r$};
          
          \draw[->] (p) edge[bend left] node {$1/0$} (q)
                        edge node[swap] {$0/1$} (r)
                        edge node[swap] {$2/1$} (e)
                    (q) edge node[above] {$0/2$, $1/2$} (p)
                        edge node {$2/2$} (r)
                    (r) edge node[below] {$0/1$, $1/1$, $2/1$} (e)
                    (e) edge[loop left] node[align=center] {$0/0$\\$1/1$\\$2/2$} (e)
          ;
        \end{tikzpicture}
      \end{center}
    \end{example}
    
    \paragraph*{Union and Power Automata.}
    The \emph{union} of two automata $\mathcal{T}_1 = (Q_1, \Sigma_1, \delta_1)$ and $\mathcal{T}_2 = (Q_2, \Sigma_2, \delta_2)$ is the automaton $\mathcal{T}_1 \cup \mathcal{T}_2 = (Q_1 \cup Q_2, \Sigma_1 \cup \Sigma_2, \delta_1 \cup \delta_2)$. The most common case of a union automaton is when the automata use the same alphabet but their state sets are disjoint. In this case, the union of two complete \SAuta is a complete \SAut.
    
    The \emph{composition} of two automata $\mathcal{T}_2 = (Q_2, \Sigma, \delta_2)$ and $\mathcal{T}_1 = (Q_1, \Sigma, \delta_1)$ over the same alphabet is $\mathcal{T}_2 \circ \mathcal{T}_1 = (Q_2 Q_1, \Sigma, \delta_2 \circ \delta_1)$ where $Q_2 Q_1 = \{ q_2 q_1 \mid q_1 \in Q_1, q_2 \in Q_2 \}$ is the Cartesian product and
    \[
      \delta_2 \circ \delta_1 = \{ \trans{p_2 p_1}{a}{c}{q_2 q_1} \mid \exists b \in \Sigma: \trans{p_1}{a}{b}{q_1} \in \delta_1, \trans{p_2}{b}{c}{q_2} \in \delta_2 \}
      \text{.}
    \]
    The composition of two complete \SAuta over the same alphabet is a complete \SAut.
    
    The \emph{$k$-th power} of an automaton $\mathcal{T} = (Q, \Sigma, \delta)$ (with $k \geq 1$) is the $k$-fold composition of $\mathcal{T}$ with itself. Any power of a complete \SAut is a complete \SAut and the action of $\bm{p} \in Q^+$ as a state sequence over $Q$ (with respect $\mathcal{T}$) is the same as the action of $\bm{p} \in Q^{|\bm{p}|}$ as a state of $\mathcal{T}^{|\bm{p}|}$, which makes the notation $\bm{p} \circ u$ unambiguous. The same holds for $\bm{p} \cdot u$ by the construction of the power automaton.
    The semigroup generated by the union of a complete \SAut with some of its powers is the same as the semigroup generated by the original automaton.
    This allows us to use the power automaton construction to make sure that any fixed state sequence acts in the same way as a single state of the automaton. For our results later on, it will be important to note here that power and union automata are computable.
    
    For the remainder of this paper, we fix an arbitrary complete \SAut $\mathcal{T} = (Q, \Sigma, \delta)$.
    \paragraph{Orbits.}
    For a language $R \subseteq Q^*$, the \emph{$R$-orbit} of a word $w \in \Sigma^* \cup \Sigma^\omega$ is 
    \[
      R \circ w = \{ \bm{r} \circ w \mid \bm{r} \in R \} \text{,}
    \]
    which is a subset of $\Sigma^{|w|}$ for $w \in \Sigma^*$ and a subset of $\Sigma^\omega$ for $w \in \Sigma^\omega$.
    The $Q^*$-orbit of $w$ is also simply called the \emph{orbit} of $w$. The \emph{orbital transducer}\footnote{The orbital transducer is basically the part reachable from $w$ in the $|w|$-th power of the dual of $\mathcal{T}$ (up to mirroring the words).} $\mathcal{T} \circ w$ of $w \in \Sigma^*$ is the complete \SAut with state set $Q^* \circ w$ and alphabet $Q$ whose transitions are given by
    \[
      \{ \trans{u}{p}{p \cdot u}{p \circ u} \mid u \in Q^* \circ w, p \in Q \} \text{.}
    \]
    We designate $w$ as the \emph{root} of $\mathcal{T} \circ w$. Since $\bm{p} \circ u$ is a left action, it is natural to write the runs of $\mathcal{T} \circ w$ from right to left. Thus, we write
    \begin{center}
      \begin{tikzpicture}[auto]
        \node (vl) {$u_\ell$};
        \node[right=of vl] (dots) {$\dots$};
        \node[right=of dots] (v1) {$u_1$};
        \node[right=of v1] (v0) {$u_0$};
        
        \path[<-] (vl) edge node {\scriptsize$p_\ell/q_\ell$} (dots)
                  (dots) edge node {\scriptsize$p_2/q_2$} (v1)
                  (v1) edge node {\scriptsize$p_1/q_1$} (v0)
        ;
      \end{tikzpicture}
    \end{center}
    to indicate that $\mathcal{T} \circ w$ contains a run from $u_0$ to $u_\ell$ with input $p_1 \dots p_\ell$ and output $q_1 \dots q_\ell$ (where $p_1, \dots, p_\ell, q_1, \dots, q_\ell \in Q$). Notice that, for $p_\ell \dots p_1 = \bm{p}$ and $q_\ell \dots q_1 = \bm{q}$, such a run means that we have $u_\ell = \bm{p} \circ u_0 = p_\ell \dots p_1 \circ u_0$ and $q_\ell \dots q_1 = \bm{q} = \bm{p} \cdot u_0 = p_\ell \dots p_1 \cdot u_0$, which justifies this somewhat reverse notation. As a short-hand notation, we also write $\lrun{u_\ell}{\bm{p}}{\bm{q}}{u_0}$ for the same run. Often, we will not be interested in the input of the run and simply write $\lrun{v}{}{\bm{q}}{u}$ to indicate that there is a run $\lrun{v}{\bm{p}}{\bm{q}}{u}$ for some $\bm{p} \in Q^+$.
    
    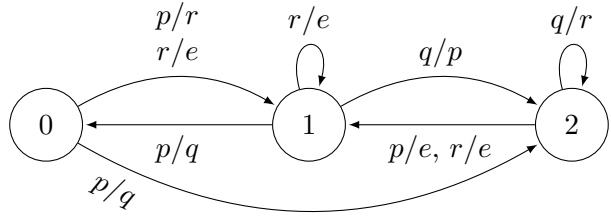
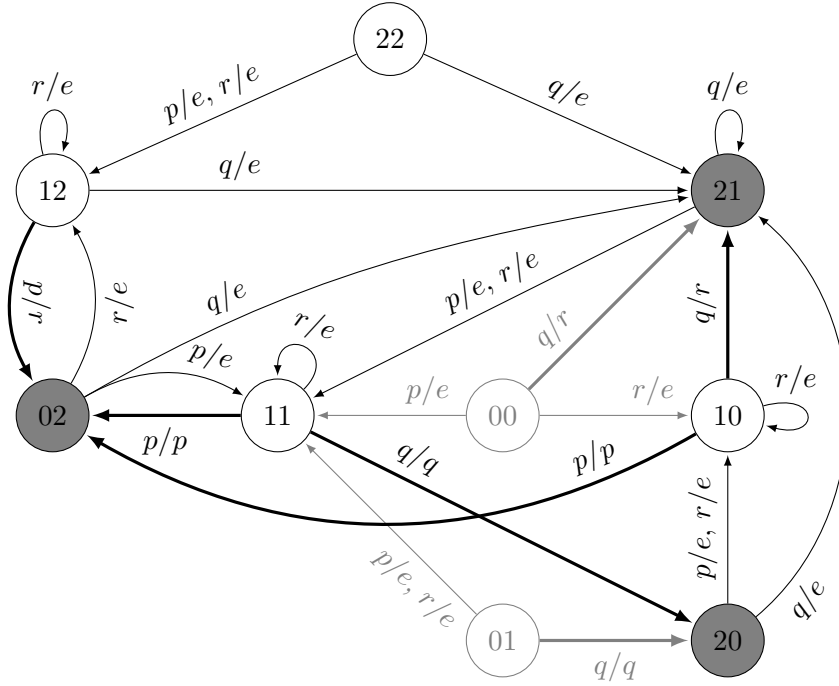
\begin{figure}
      \begin{subfigure}{\linewidth}\centering
        \begin{tikzpicture}[auto, shorten >=1pt, >=latex, node distance=2.5cm]
          \node[state] (0) {$0$};
          \node[state, right=of 0] (1) {$1$};
          \node[state, right=of 1] (2) {$2$};
          
          \draw[->] (0) edge[bend left] node[align=center] {$p/r$\\$r/e$} (1)
                        edge[bend right] node[sloped, below, pos=0.1] {$p/q$} (2)
                    (1) edge[loop above] node {$r/e$} (1)
                        edge node {$p/q$} (0)
                        edge[bend left] node {$q/p$} (2)
                    (2) edge node[below] {$p/e$, $r/e$} (1)
                        edge[loop above] node {$q/r$} (2)
          ;
        \end{tikzpicture}
        \caption{The orbital transducer $\mathcal{T} \circ 0 = \mathcal{T} \circ 1 = \mathcal{T} \circ 2$ (which is equal to the dual automaton); the $e/e$-self-loops at every state are not drawn.}
      \end{subfigure}\\%
      \begin{subfigure}{\linewidth}\centering
        \begin{tikzpicture}[auto, shorten >=1pt, >=latex, node distance=2cm]
          \tikzstyle{active} = [fill=gray]
          \tikzstyle{eedge} = []
          \tikzstyle{aedge} = [very thick]
          
          \node[state, active] (02) {$02$};
          \node[state, right=of 02] (11) {$11$};
          \node[state, right=of 11, gray] (00) {$00$};
          \node[state, right=of 00] (10) {$10$};
          \node[state, above=of 10, active] (21) {$21$};
          \node[state, below=of 00, gray] (01) {$01$};
          \node[state, below=of 10, active] (20) {$20$};
          \node[state, above=of 02] (12) {$12$};
          \node[state, anchor=base] at ($(12.base)!0.5!(21.base)+(0pt, 2cm)$) (22) {$22$};
          
          \draw[->]
            (00) edge[eedge, gray] node[sloped, above, pos=0.25] {$p/e$} (11)
                 edge[aedge, gray] node[sloped, above, pos=0.25] {$q/r$} (21)
                 edge[eedge, gray] node[above, pos=0.75] {$r/e$} (10)
            (01) edge[eedge, gray] node[sloped, below, pos=0.25] {$p/e$, $r/e$} (11)
                 edge[aedge, gray] node[sloped, below] {$q/q$} (20)
            (02) edge[eedge, bend left] node[sloped, above, pos=0.75] {$p/e$} (11)
                 edge[eedge, out=30, in=-170] node[sloped, above, pos=0.25] {$q/e$} (21)
                 edge[eedge, bend right] node[sloped, below] {$r/e$} (12)
            (10) edge[aedge, out=-150, in=-30] node[sloped, above, pos=0.15] {$p/p$} (02)
                 edge[aedge] node[sloped, above] {$q/r$} (21)
                 edge[eedge, loop right] node[above, pos=0.25] {$r/e$} (10)
            (11) edge[aedge] node[below] {$p/p$} (02)
                 edge[aedge] node[sloped, above, pos=0.25] {$q/q$} (20)
                 edge[eedge, out=45, in=90, looseness=5] node[sloped] {$r/e$} (11)
            (12) edge[aedge, bend right] node[sloped] {$p/r$} (02)
                 edge[eedge] node[pos=0.25, above] {$q/e$} (21)
                 edge[eedge, loop above] node {$r/e$} (12)
            (20) edge[eedge] node[sloped, above] {$p/e$, $r/e$} (10)
                 edge[eedge, bend right=50] node[sloped, below, pos=0.1] {$q/e$} (21)
            (21) edge[eedge] node[sloped, above] {$p/e$, $r/e$} (11)
                 edge[eedge, loop above] node {$q/e$} (21)
            (22) edge[eedge] node[sloped, above] {$p/e$, $r/e$} (12)
                 edge[eedge] node[sloped, above] {$q/e$} (21) %!!!!!!!!!!!
          ;
        \end{tikzpicture}
        \caption{The orbital transducer $\mathcal{T} \circ 22$;
          the $e/e$-self-loops at every state are not drawn.
          Active words are marked in gray.
          Edges with output different to $e$ are drawn in bold.
          The states $00$ and $01$ are not reachable from $22$ and, thus, not part of the orbital transducers.
          They are drawn here nevertheless to also illustrate $\mathcal{T} \circ 00$ and $\mathcal{T} \circ 01$ (as well as all other orbital transducers for words of length $2$).
        }
      \end{subfigure}
      \caption{Orbital transducers for the automaton from the running example.}\label{fig:orbTransRunning}
    \end{figure}
    \begin{example}[Running Example]
      Recall $\mathcal{T}$ from \autoref{ex:runningAut}.
      The orbital transducers $\mathcal{T} \circ 0$ (which is equal to $\mathcal{T} \circ 1$, $\mathcal{T} \circ 2$ and the dual automaton) and the orbital transducer $\mathcal{T} \circ 22$ are shown in \autoref{fig:orbTransRunning}.
      Observe that, reading $22$ from state $p$ yields output $12$ and we end in state $e$; this yields the edge $\trans{22}{p}{e}{12}$ in $\mathcal{T} \circ 22$.
      If we read $10$ in state $p$, we obtain the output $02$ and return to state $p$ in the end; this yields the edge $\trans{10}{p}{p}{02}$ in $\mathcal{T} \circ 22$.
      All other edges arise in the same manner.
    \end{example}
    
    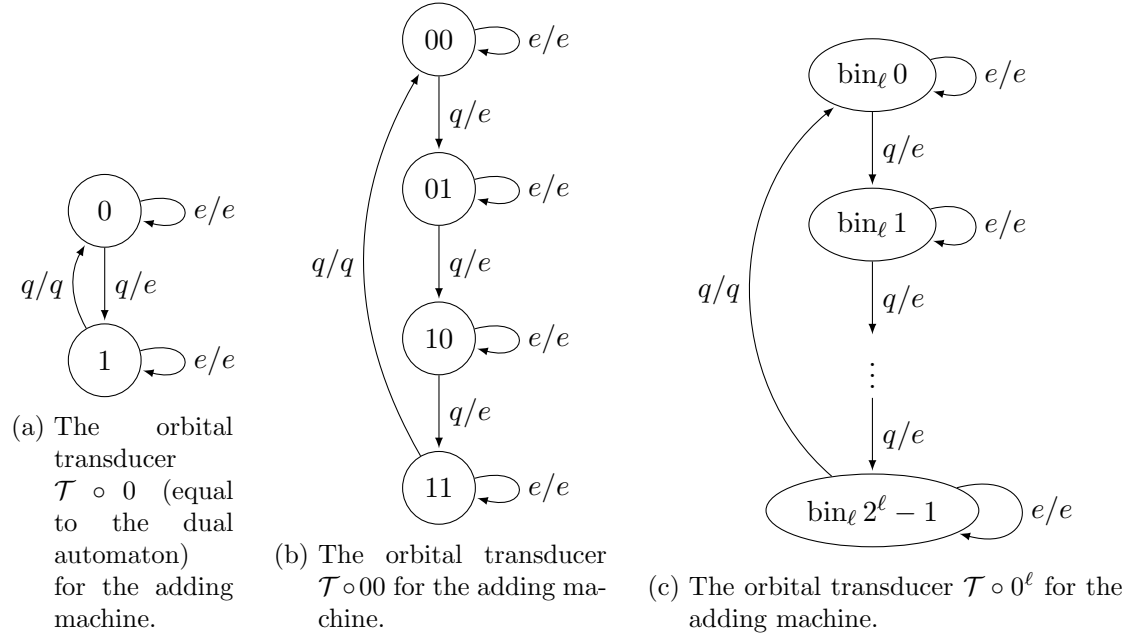
\begin{figure}
      \begin{subfigure}{0.2\linewidth}\centering
        \begin{tikzpicture}[auto, shorten >=1pt, >=latex]
          \node[state] (0) {$0$};
          \node[state, below=of 0] (1) {$1$};
          
          \draw[->] (0) edge[loop right] node {$e/e$} (0)
                        edge node {$q/e$} (1)
                    (1) edge[loop right] node {$e/e$} (1)
                        edge[bend left] node {$q/q$} (0)
          ;
        \end{tikzpicture}
        \caption{The orbital transducer $\mathcal{T} \circ 0$ (equal to the dual automaton) for the adding machine.}
      \end{subfigure}\hspace*{\fill}%
      \begin{subfigure}{0.3\linewidth}\centering
        \begin{tikzpicture}[auto, shorten >=1pt, >=latex]
          \node[state] (00) {$00$};
          \node[state, below=of 00] (01) {$01$};
          \node[state, below=of 01] (10) {$10$};
          \node[state, below=of 10] (11) {$11$};
          
          \draw[->] (00) edge[loop right] node {$e/e$} (00)
                         edge node {$q/e$} (01)
                    (01) edge[loop right] node {$e/e$} (01)
                         edge node {$q/e$} (10)
                    (10) edge[loop right] node {$e/e$} (10)
                         edge node {$q/e$} (11)
                    (11) edge[loop right] node {$e/e$} (11)
                         edge[bend left] node {$q/q$} (00)
          ;
        \end{tikzpicture}
        \caption{The orbital transducer $\mathcal{T} \circ 00$ for the adding machine.}
      \end{subfigure}\hspace*{\fill}%
      \begin{subfigure}{\dimexpr0.45\linewidth-2ex}\centering
        \begin{tikzpicture}[auto, shorten >=1pt, >=latex]
          \node[state, ellipse] (0) {$\operatorname{bin}_\ell 0$};
          \node[state, below=of 0, ellipse] (1) {$\operatorname{bin}_\ell 1$};
          \node[below=of 1] (dots) {$\vdots$};
          \node[state, below=of dots, ellipse] (11) {$\operatorname{bin}_\ell 2^\ell - 1$};
          
          \draw[->] (0) edge[loop right, looseness=5] node {$e/e$} (0)
                        edge node {$q/e$} (1)
                    (1) edge[loop right, looseness=5] node {$e/e$} (1)
                        edge node {$q/e$} (dots)
                    (dots) edge node {$q/e$} (11)
                    (11) edge[loop right, looseness=5] node {$e/e$} (11)
                         edge[bend left=50] node {$q/q$} (0)
          ;
        \end{tikzpicture}
        \caption{The orbital transducer $\mathcal{T} \circ 0^\ell$ for the adding machine.}
      \end{subfigure}%
      \caption{Orbital transducers for the adding machine.}\label{fig:orbTransAddingMachine}
    \end{figure}
    \begin{example}[Orbital Transducers of the Adding Machine]
      Recall the adding machine $\mathcal{T}$ from \autoref{ex:addingMachine}.
      Also recall the notation $\operatorname{bin}_\ell n$ from there.
      The orbital transducers $\mathcal{T} \circ 0$, $\mathcal{T} \circ 00$ and a general depiction of $\mathcal{T} \circ 0^\ell$ may be found in \autoref{fig:orbTransAddingMachine}.
      The orbital transducer $\mathcal{T} \circ 0$ coincides with what is commonly referred to as the \emph{dual automaton} of $\mathcal{T}$.
      Observe that, if we start reading $1^\ell$ in $q$, we output $0^\ell$ while remaining in $q$ the entire time.
      This yields the $q/q$-edge from $1^\ell$ to $0^\ell$.
      If we read any other word (i.\,e.\ one that contains at least one $0$), we end in the state $e$.
      This is the reason why any other $q$-edge has output $e$.
    \end{example}
    
    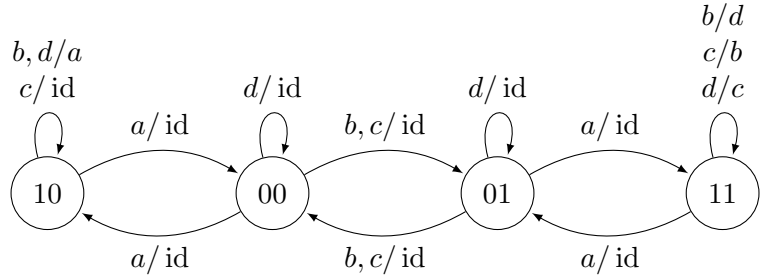
\begin{figure}\centering
      \begin{tikzpicture}[auto, shorten >=1pt, >=latex, node distance=2cm]
        \node[state] (10) {$10$};
        \node[state, right=of 10] (00) {$00$};
        \node[state, right=of 00] (01) {$01$};
        \node[state, right=of 01] (11) {$11$};
        
        \draw[->] (10) edge[loop above] node[align=center] {$b, d/a$\\$c/\operatorname{id}$} (10)
                       edge[bend left] node {$a/\operatorname{id}$} (00)
                  (00) edge[loop above] node {$d/\operatorname{id}$} (00)
                       edge[bend left] node {$a/\operatorname{id}$} (10)
                       edge[bend left] node {$b, c/\operatorname{id}$} (01)
                  (01) edge[loop above] node {$d/\operatorname{id}$} (01)
                       edge[bend left] node {$a/\operatorname{id}$} (11)
                       edge[bend left] node {$b, c/\operatorname{id}$} (00)
                  (11) edge[loop above] node[align=center] {$b/d$\\$c/b$\\$d/c$} (11)
                        edge[bend left] node {$a/\operatorname{id}$} (01)
        ;
      \end{tikzpicture}
      \caption{The orbital transducer $\mathcal{T} \circ 00$ for Grigorchuk's group.}\label{fig:orbTransGrig}
    \end{figure}
    \begin{example}[Orbital Transducers for Grigorchuk's Group]
      Recall $\mathcal{T}$ generating Grigorchuk's group from \autoref{ex:grigorchuk}.
      The orbital transducer $\mathcal{T} \circ 00$ is drawn in \autoref{fig:orbTransGrig}
    \end{example}
    
    \paragraph{The Product $\mathcal{T} \circ w \times R$.}
    For some language $R \subseteq Q^*$, the \emph{product} $\mathcal{T} \circ w \times R$ is a possibly infinite \SAut over the alphabet $Q$. Its states are of the form $(u, C)$ where $u \in Q^* \circ w$ and $C$ is a class of $R$ and we have a transition $\ltrans{(v, \mathscr{C}(\bm{r}p))}{p}{q}{(u, \mathscr{C}(\bm{r}))}$ if we have a transition $\ltrans{v}{p}{q}{u}$ in $\mathcal{T} \circ w$. We restrict $\mathcal{T} \circ w \times R$ to the part reachable from its \emph{root} $(w, \mathscr{C}(\varepsilon))$. We use the same conventions to write transitions and runs in $\mathcal{T} \circ w \times R$ as in $\mathcal{T} \circ w$. In particular, we omit the input label if we are not interested in it.
    
    Note that we have a run $\lrun{(u, C)}{\bm{p}}{\bm{q}}{(w, \mathscr{C}(\varepsilon))}$ in $\mathcal{T} \circ w \times R$ if and only if we have $u = \bm{p} \circ w$, $\bm{q} = \bm{p} \cdot w$ and $C = \mathscr{C}(\bm{p})$. Furthermore, $\mathcal{T} \circ w \times R$ is finite if $R$ is regular.
    
    \paragraph{Expandability.}
    A finite word $w \in \Sigma^*$ is \emph{$R$-expandable} for $R \subseteq Q^*$ if there is some $x \in \Sigma^*$ with $|R \circ w| < |R \circ wx|$. In this case, $x$ is said to \emph{$R$-expand} $w$.
    \begin{remark*}
      The notion of expandability was introduced in \cite{expandabilityPart} where it was shown that it is decidable whether a given word is expandable with respect to the action of a given \SAut and that there is a bound on the length of the expanding suffix.
    \end{remark*}
  \end{section}
  \begin{section}{Generalized Activity}
    We say that a subset $S \subseteq Q$ of states of $\mathcal{T}$ is \emph{closed} (under the dual action) if we have
    \[
      S \cdot \Sigma^* = \{ s \cdot u \mid s \in S, u \in \Sigma^* \} \subseteq S \text{.}
    \]
    For such a closed subset $S$, we may consider the restriction $\delta|_S$ of $\delta$ into a relation $\delta|_S \subseteq S \times \Sigma \times \Sigma \times S$, which yields a complete \SAut $\mathcal{T}|_S = (S, \Sigma, \delta|_S)$. Notice that the action $\bm{s} \circ u$ of some $\bm{s} \in S^+$ is the same when we consider it with respect to $\mathcal{T}$ or with respect to $\mathcal{T}|_S$. This justifies that we define the semigroup generated by $\mathcal{T}|_S$ as $S^+/{=_{\mathcal{T}}}$ (which is the subsemigroup generated by $S$ in the semigroup generated by $\mathcal{T}$).
    
    We are mostly interested in the case that this generated semigroup is finite. Therefore, we fix an arbitrary closed subset $S \subseteq Q$ such that the generated semigroup $S^+/{=_{\mathcal{T}}}$ is finite.
    
    \begin{remark}\label{rem:SIsSubsetOfStates}
      Only considering subsets of states here is not a real restriction: If we have that $S \subseteq Q^+$ generates a finite subsemigroup in $\mathscr{S}(\mathcal{T})$, we may assume that $S$ is finite (by choosing one representative for each of the finitely many elements in the subsemigroup). As briefly mentioned above, we may then pass to a union of suitable power automata of $\mathcal{T}$, which allows us to finally assume $S \subseteq Q$.
    \end{remark}
    
    For a state sequence $\bm{p} \in Q^+$, we define its set of \emph{$S$-active} words of \emph{length} $n \in \mathbb{N}$ as
    \[
      A_{\bm{p}}(n) = \{ v \mid \exists u \in \Sigma^n: \bm{p} \circ u = v, \bm{p} \cdot u \not\in_\mathcal{T} S^+ \}
    \]
    and the set of \emph{$S$-active} words as $A_{\bm{p}} = \bigcup_{n \geq 0} A_{\bm{p}}(n)$. The \emph{$S$-activity} of $\bm{p}$ is the growth of the function $n \mapsto \alpha_{\bm{p}}(n) = |A_{\bm{p}}(n)|$.
    
    Please note that, instead of considering $\bm{p}$ as a state sequence over $Q$, we may also consider it as a single state of $\mathcal{T}^{|\bm{p}|}$ and obtain the same sets of active words.
    
    With regard to orbital transducers, a word $v \in \Sigma^n$ is $S$-active (i.\,e.\ in $A_{\bm{p}}(n)$) if and only if there is some $u \in \Sigma^n$ with a run $\lrun{v}{\bm{p}}{\bm{q}}{u}$ in $\mathcal{T} \circ u$ such that $\bm{q} \not\in_{\mathcal{T}} S^+$. In other words, we collect all words with an in-going run in any orbital transducer $\mathcal{T} \circ w$ whose output is not labeled by an element contained in the subsemigroup generated by $S$.
    
    \begin{remark*}
      Our notion of $S$-activity generalizes the activity notion for automaton monoids introduced in \cite{bartholdi2020hierarchy}. An automaton monoid is an automaton semigroup generated by a complete \SAut $\mathcal{T} = (Q, \Sigma, \delta)$ containing an identity state $e \in Q$ (i.\,e.\ we have $\trans{e}{a}{a}{e} \in \delta$ for all $a \in \Sigma$).\footnote{Please note that \cite{bartholdi2020hierarchy} uses the slightly misleading term \enquote{automaton semigroup} for what we call an automaton monoid here.} For such an automaton, the set $\{ e \}$ is clearly closed and its generated subsemigroup is the trivial monoid (and, thus, in particular finite). The activity defined in \cite{bartholdi2020hierarchy} is then precisely the $\{ e \}$-activity in our terminology. In fact, \cite{bartholdi2020hierarchy} already contains an extension of the activity notion where the set $S$ is given by the NoCyWEx part of the generating automaton (which always generates a finite subsemigroup).
      
      The authors of \cite{bartholdi2020hierarchy} use NoCyWEx as an acronym for \enquote{no cycles with exit}.
      We do not strictly need the NoCyWEx notion for (sub)automata for our results but we will compare our generalization to it.
      Therefore, it seems worthwhile to give a definition here. We invite the interested reader to refer to \cite{bartholdi2020hierarchy} for more details and references.
      An automaton (or subautomaton over the same alphabet) is called \emph{NoCyWEx} if, for every state $q \in Q$ on a cycle, all transitions starting in $q$ lead to the same state, i.\,e.\ $\{ q \cdot a \mid a \in \Sigma \}$ contains precisely one element for all $q \in Q$ that lie on a cycle (which are those $q \in Q$ with $q \cdot w = q$ for some $w \in \Sigma^+$).
      
      The activity notion for automaton monoids from \cite{bartholdi2020hierarchy} is, in turn, a generalization of the original activity notion for automaton groups (which are generated by invertible automata) introduced by Sidki \cite{sidki2000automorphisms}. The main difference is that, in the group case, we may use the input words to define the set of active words, i.\,e.\ we could let
      \[
        A_{\bm{p}}(n) = \{ u \in \Sigma^n \mid \bm{p} \circ u \neq_\mathcal{T} e \}
      \]
      (where $e$ is the identity state). The straight-forward generalization of this definition to automaton monoids is not well-behaved (and we refer the reader to \cite{bartholdi2020hierarchy} for more details on this).
    \end{remark*}
    
    \begin{example}[Running Example]
      Recall the automaton $\mathcal{T}$ from \autoref{ex:runningAut}.
      We let $S = \{ e \}$ and observe that $S^+ / {=_\mathcal{T}}$ is the trivial monoid (and, thus, in particular, finite).
      The word $02$ is $S$-active (of length $2$) for the state $p$:
      We have $p \circ 10 = 02$ and, after reading $10$ from $p$, we end in state $p$ again.
      Since we have $p \not\in_\mathcal{T} e^+$ (as $p \neq_\mathcal{T} e$), this shows that $02$ is indeed $\{ e \}$-active.
      
      We may also observe that $02$ is $\{ e \}$-active for $p$ from the orbital transducer $\mathcal{T} \circ 22$ (see \autoref{fig:orbTransRunning}):
      $02$ has an incoming $p/p$-edge (from $10$ and also from $11$).
      
      On the other hand, $11$ is \textbf{not} $\{ e \}$-active for $p$:
      While there are input words (e.\,g.\ $21$ or $00$) yielding $11$ as the output when starting in $p$ (e.\,g.\ $p \circ 21 = 11$ and $p \circ 00 = 11$), we always end in state $e$ after reading any of them from $p$ (e.\,g.\ $p \cdot 21, p \cdot 00 \in_\mathcal{T} e^+$).
      
      In fact, $11$ is not $\{ e \}$-active for any state of the automaton, which we may see by observing that the orbital transducer $\mathcal{T} \circ 22$ (from \autoref{fig:orbTransRunning}, which contains all words of length $2$) does not have any incoming transitions at $11$ with output different to $e$.
      Similarly, we may also observe that $00$, $01$, $10$, $12$ and $22$ are also not $\{ e \}$-active for any state of the automaton.
    \end{example}

    Our generalization of the activity notion maintains many desirable properties of the original notion(s). First, it remains subadditive (compare to \cite[Lemma~5]{bartholdi2020hierarchy}).
    \begin{fact}\label{fct:subadditive}
      For $\bm{p}, \bm{q} \in Q^+$, we have $\alpha_{\bm{q} \bm{p}}(n) \leq \alpha_{\bm{q}}(n) + \alpha_{\bm{p}}(n)$ for all $n \geq 0$.
    \end{fact}
    \begin{proof}
      We show $A_{\bm{qp}}(n) \subseteq A_{\bm{q}}(n) \cup \bm{q} \circ A_{\bm{p}}(n)$ where $\bm{q} \circ A_{\bm{p}}(n) = \{ \bm{q} \circ v \mid v \in A_{\bm{p}}(n) \}$ and, thus, $| \bm{q} \circ A_{\bm{p}}(n) | \leq | A_{\bm{p}}(n) | = \alpha_{\bm{p}}(n)$.
      
      Suppose we have $w \in A_{\bm{qp}}(n)$. By definition, there is some $u \in \Sigma^n$ with $\bm{qp} \circ u = w$ and $\bm{qp} \not\in_{\mathcal{T}} S^+$. We cannot have $\bm{p} \cdot u \in_{\mathcal{T}} S^+$ and $\bm{q} \cdot v \in_{\mathcal{T}} S^+$ for $v = \bm{p} \circ u$ since this would imply $\bm{qp} \cdot u = (\bm{q} \cdot v) (\bm{p} \cdot u) \in_{\mathcal{T}} S^+$. Thus, we have $\bm{p} \cdot u \not\in_{\mathcal{T}} S^+$ or $\bm{q} \cdot v \not\in_{\mathcal{T}} S^+$. In the first case, we have $v = \bm{p} \circ u \in A_{\bm{p}}(n)$ and, thus, $w = \bm{q} \circ v \in \bm{q} \circ A_{\bm{p}}(n)$. In the second case, we have $w \in A_{\bm{q}}(n)$ directly by definition.
    \end{proof}
    
    Next, the set of $S$-active words is regular.
    \begin{fact}
      The set of $S$-active words $A_{\bm{p}}$ is regular for all state sequences $\bm{p} \in Q^+$.
    \end{fact}
    \begin{proof}
      We will describe how to obtain (in fact, compute) a non-deterministic finite acceptor (NFA) recognizing $A_{\bm{p}}$. We refer the reader to \cite{bartholdi2020hierarchy}, \cite{waechter2024word} or a standard textbook on formal language theory (such as \cite{hopcroft1979introduction}) for proper definitions of NFAs and how they relate to our definition of regular languages.
      
      Recall that we may consider $\bm{p}$ as a single state of the power automaton $\mathcal{T}^{|\bm{p}|}$ (and still obtain the same set of active words). By replacing $\mathcal{T}$ with $\mathcal{T}^{|\bm{p}|}$, we may, thus, without loss of generality assume $\bm{p} = p \in Q$. We mark all states $q$ in $\mathcal{T}$ with $q \not\in_{\mathcal{T}} S^+$ as accepting and $p$ as initial. Afterwards, we drop the inputs from all transitions and obtain an NFA. By construction, we obtain that the words accepted by this are precisely the $S$-active ones.
    \end{proof}
    
    The \emph{growth} of a language $L \subseteq \Sigma^*$ is the growth of the function $n \mapsto \gamma_L(n) = |L \cap \Sigma^n|$. Thus, by definition, $\alpha_{\bm{p}}$ is the growth function of the language of $S$-active words. It is well-known that a regular language either grows polynomially of exponentially (see e.\,g.\ \cite[Subsection~6.5.1]{waechter2024word}). Here, \emph{exponential growth} means that there is some $r > 1$ such that $\gamma_L(n) \geq r^n$ for infinitely many $n$ (while, of course, $\gamma_L(n) \leq |\Sigma|^n$ for all $n$) and \emph{polynomial growth of degree exactly} $d \geq 0$ means that there is some $d$ and polynomials $p(n)$ and $q(n)$ of degree $d$ with $\gamma_L(n) \geq p(n)$ for infinitely many $n$ and $\gamma_L(n) \leq q(n)$ for all $n$. The special case of \emph{polynomial growth of degree $-\infty$} occurs if $\gamma_L$ eventually becomes the constant zero function (which happens if and only if $L$ is a finite language). We say that a language has polynomial growth of degree \emph{at most} $d$ if it has polynomial growth of degree exactly $d'$ for some $d' \leq d$.
    
    Now, a language has polynomial growth if it has polynomial growth of degree $d$ for some $d \in \{ -\infty \} \cup \mathbb{N}$. Since the $S$-activity of the state sequence $\bm{p}$ is the growth of its language of $S$-active words, this immediately induces definitions of \emph{polynomial} or \emph{exponential} $S$-activity and we obtain:
    \begin{fact}
      The $S$-activity of a state sequence is either polynomial or exponential.
    \end{fact}
    
    One important special case of polynomial $S$-activity occurs when there is a constant bounding $\alpha_{\bm{p}}(n)$ for all $n$. In this case, the $S$-activity of $\bm{p}$ is called \emph{bounded} and this is what we will mostly be concerned with in this paper. If $\alpha_{\bm{p}}(n)$ eventually becomes the constant zero function (i.\,e.\ if the language of $S$-active words is finite), we say that $\bm{p}$ is \emph{$S$-finitary}.

    \paragraph*{Extension to Automata.}
    The subadditivity stated in \autoref{fct:subadditive} ensure that, if we take two state sequences $\bm{q}$ and $\bm{p}$ of polynomial activity (or bounded activity), then their concatenation $\bm{q} \bm{p}$ remains to have polynomial (bounded) activity. This justifies defining the activity of the entire \SAut $\mathcal{T} = (Q, \Sigma, \delta)$ as the activity of the states. More precisely, we define the set of \emph{$S$-active words} of length $n$ for the whole automaton as
    \[
      A(n) = \bigcup_{q \in Q} A_q(n)
    \]
    and the \emph{$S$-activity} of $\mathcal{T}$ as the growth of the function $n \mapsto \alpha(n) = |A(n)|$. Clearly, we have $\alpha_p(n) \leq \alpha(n) \leq \sum_{q \in Q} \alpha_q(n)$ for all $n \in \mathbb{N}$ and $p \in Q$. Thus, and by \autoref{fct:subadditive}, we also have that the $S$-activity of any state sequence $\bm{p} \in Q^+$ is polynomial with degree at most $d$ if the $S$-activity of $\mathcal{T}$ is.
    
    As with state sequences, we say that $\mathcal{T}$ is of \emph{bounded} $S$-activity if there is some constant bounding the $S$-activity of $\mathcal{T}$ and we say that $\mathcal{T}$ is \emph{$S$-finitary} if all states of $\mathcal{T}$ are.
    
    \begin{example}[Running Example]
      Again, let us consider the automaton from \autoref{ex:runningAut}.
      We can fully describe all the $\{ e \}$-active words for all states of the automaton.
      Using standard notation from formal language theory, we have:
      \begin{align*}
        A_p(n) &=
          \begin{cases}
            \{ \left( 02 \right)^{\frac{n}{2}} \} & n \text{ even}\\
            \left( 02 \right)^{\frac{n}{2} - 1} \{ 0, 1 \} & n \text{ odd}
          \end{cases}\\
        A_q(n) &=
          \begin{cases}
            \{ (20)^\frac{n}{2}, (20)^{\frac{n}{2} - 1}21 \} & n \text{ even}\\
            \{ (20)^\frac{n - 1}{2}2 \} & n \text{ odd}
          \end{cases}\\
        A_r(n) &= \emptyset
      \intertext{Accordingly, we have for the entire automaton}
        A(n) &=
          \begin{cases}
            \{ (02)^\frac{n}{2}, (20)^\frac{n}{2}, (20)^{\frac{n}{2} - 1}21 \} & n \text{ even}\\
            (02)^\frac{n - 1}{2}\{ 0, 1 \} \cup \{ (20)^\frac{n - 1}{2} 2 \} & n \text{ odd}
          \end{cases}
      \end{align*}
      and $|A(n)| \leq 3$ for all $n \geq 0$.
      This shows that the automaton is of bounded $\{ e \}$-activity.
    \end{example}
    
    \begin{example}
      The adding machine from \autoref{ex:addingMachine} is well known to have bounded $\{ e \}$-activity; the only $\{ e \}$-active words are those in $0^*$.
      No state in the automaton from \autoref{ex:finiteMonoid} has any $\{ e \}$-active words; it, therefore, is $\{ e \}$-finitary.
      
      The construction in \autoref{ex:finiteSemigroup} to generate a finite semigroup $S$ does not contain an identity state if $S$ is not a monoid and, thus, has trivially exponential activity in the sense of \cite{bartholdi2020hierarchy} in that case.
      The generating automaton is NoCyWEx in the sense of \cite{bartholdi2020hierarchy} and, thus, also finitary in this sense.
      
      However, there are also non-NoCyWEx automata generating finite semigroups. For example, the complete \SAut
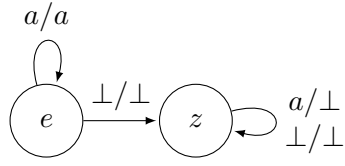
      \begin{center}
        \begin{tikzpicture}[auto, shorten >=1pt, >=latex]
          \node[state] (e) {$e$};
          \node[state, right=of e] (z) {$z$};
          
          \path[->] (e) edge[loop above] node {$a/a$} (e)
                        edge node {$\bot/\bot$} (z)
                    (z) edge[loop right, align=center] node {$a/\bot$\\$\bot/\bot$} (z)
          ;
                       
        \end{tikzpicture}
      \end{center}
      with alphabet $\{ a, \bot \}$ generates the monoid $U_1 = \{ e, z \}$ with $e^2 = e$ and $z^2 = ez = ze = z$. As $z$ is not the neutral element, this automaton has exponential activity in the sense of \cite{bartholdi2020hierarchy} and it has bounded activity (but is not finitary!) in the NoCyWEx-sense of \cite{bartholdi2020hierarchy}.
      
      We may even go further. As an example, consider the complete \SAut $\mathcal{T} = (\{ q, e, z \}, \{ 0, 1, a, \bot \}, \delta)$
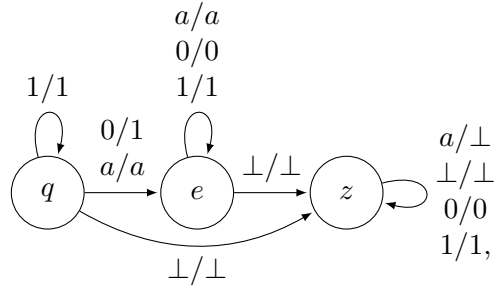
      \begin{center}
        \begin{tikzpicture}[auto, shorten >=1pt, >=latex]
          \node[state] (e) {$e$};
          \node[state, right=of e] (z) {$z$};
          \node[state, left=of e] (q) {$q$};
          
          \path[->] (e) edge[loop above] node[align=center] {$a/a$\\$0/0$\\$1/1$} (e)
                        edge node {$\bot/\bot$} (z)
                    (z) edge[loop right, align=center] node {$a/\bot$\\$\bot/\bot$\\$0/0$\\$1/1$,} (z)
                    (q) edge[loop above] node {$1/1$} (q)
                        edge node[align=center] {$0/1$\\$a/a$} (e)
                        edge[bend right] node[swap] {$\bot/\bot$} (z)
          ;
          
        \end{tikzpicture}
      \end{center}
      which combines the above example (with additional $0/0$- and $1/1$-self-loops at all states that do not change the generated semigroup) with the adding machine from \autoref{ex:addingMachine}.
      Using an induction on $|u|$ for $u \in \{ 0, 1, a, \bot \}^*$ and an exhaustive calculation, one may easily see that all states commute in the generated semigroup and that we have $qe =_{\mathcal{T}} q$, which turns $e$ into a neutral element.
      Next, we have
      \begin{align*}
        q^i \circ 0^\ell = q^i e \circ 0^\ell = q^i z \circ 0^\ell = \operatorname{bin}_\ell i
      \end{align*}
      (for $\ell$ sufficiently large), which shows $q^i s \neq_{\mathcal{T}} q^j t$ for all $i \neq j$ and $s, t \in \{ e, z \}$. Finally, we clearly have $q^i e \neq_{\mathcal{T}} q^i z$ for all $i \geq 0$ and, thus, that the semigroup generated by $\mathcal{T}$ is (isomorphic to) $q^* \times U_1$.
      
      The interesting part about this last example $\mathcal{T}$ is that the automaton has exponential activity both with respect to the identity function (and even its neutral element $e$) as well as to its NoCyWEx part (which only contains of $z$). Nevertheless, it has bounded $\{ e, z \}$-activity (where the subsemigroup generated by $e$ and $z$ is isomorphic to $U_1$).
    \end{example}
    
    \begin{figure}\centering
      \begin{tikzpicture}[auto, shorten >=1pt, >=latex]
        \node[state] (p) {$p$};
        \node[state, right=4cm of p] (q) {$q$};
        \node[state, below=of p] (e) {$e$};
        \node[state, below=of q] (r) {$r$};
        
        \draw[->] (p) edge[bend left] node {$1/0$} (q)
                      edge node[swap] {$0/1$} (r)
                      edge node[swap] {$2/1$} (e)
                  (q) edge node[above] {$0/0$, $1/2$, $2/0$} (p)
                  (r) edge node[below] {$0/1$, $1/1$, $2/0$} (e)
                  (e) edge[loop left] node[align=center] {$0/0$\\$1/1$\\$2/2$} (e)
        ;
      \end{tikzpicture}
      \caption{A complete \SAut with exponential activity.}\label{fig:modifiedAut}
    \end{figure}
    \begin{example}[Exponential Activity]
      By slightly modifying the automaton from \autoref{ex:runningAut}, we may obtain the automaton depicted in \autoref{fig:modifiedAut}.
      Here, the $\{ e \}$-active words for the states are:
      \begin{align*}
        A_p(n) &=
        \begin{cases}
          \left( 0 \{ 0, 2 \} \right)^\frac{n}{2} & n \text{ even}\\
          \left( 0 \{ 0, 2 \} \right)^\frac{n - 1}{2} \{ 0, 1 \} & n \text{ odd}
        \end{cases}\\
        A_q(n) &=
        \begin{cases}
          \left( \{ 0, 2 \} 0 \right)^\frac{n}{2} \cup \left( \{ 0, 2 \} 0 \right)^{\frac{n}{2} - 1} \{ 0, 2 \} 1 & n \text{ even}\\
          \left( \{ 0, 2 \} 0 \right)^{\frac{n - 1}{2}} \{ 0, 2 \} & n \text{ odd}
        \end{cases}\\
        A_r(n) &= \emptyset
      \end{align*}
      In particular, the automaton has exponential $\{ e \}$-activity.
    \end{example}
    
    \begin{remark}
      The isomorphism problem for automaton groups and, thus, for automaton semigroups is undecidable (this follows from \cite{sunic2012conjugacy} but, in the semigroup case, also from the direct construction in \cite{dangeli2024freeness}). However, the (uniform) word problem is decidable (in fact, it is $\ComplexityClass{PSpace}$-complete \cite{dangeli2017complexity,waechter2023automaton}, see also \cite{waechter2024word}). Thus, for any fixed finite semigroup $T$, we may test whether there is a (closed) subset $P$ of states whose generated subsemigroup is isomorphic to $T$:
      We enumerate all elements of the subsemigroup (for example, using breath-first search in the Cayley graph) until we have found all elements or more than $|T|$-many elements. In the second case, the subsemigroup cannot be isormorphic to $T$ and, in the first case, we may test the isomorphism between finite semigroups.
      This allows us to consider our notion of $S$-activity in a more abstract sense where, instead of fixing a subset of states, we fix some finite semigroup $T$ and then try to find a closed subset $S$ of states generating $T$. If there is such an $S$, we may use it to base the activity on; if there is no such $S$, we set $S = \emptyset$ and obtain exponential activity.
    \end{remark}
  \end{section}

  \begin{section}{Expansion Relation}
    Recall that we fixed some complete \SAut $\mathcal{T} = (Q, \Sigma, \delta)$ and some closed subset $S \subseteq Q$ such that $S^+/{=_{\mathcal{T}}}$ is finite.
    Without loss of generality, we may assume that for every $\bm{t} \in S^+$, there is exactly one $s \in S$ with $\bm{t} =_{\mathcal{T}} s$: since there are only finitely many equivalence classes under $=_{\mathcal{T}}$, we may replace $\mathcal{T}$ by a finite union of suitable powers of $\mathcal{T}$ and then minimize the automaton to avoid multiple states with the same action.\footnote{%
      For our algorithmic results, it is important to note that these powers (and the minimization) can be computed from $\mathcal{T}$ and the subset $S$ if the semigroup $S^+/{=_\mathcal{T}}$ is finite.
      For more information on how to minimize an automaton, see e.\,g.\ \cite[Subsection 2.3]{klimann2012implementing} or \cite[Appendix~A.5]{linz2016introduction}.%
    }
    
    We write $P = Q \setminus S$ for the states outside the closed subset and consider the free product $(S^+/{=_{\mathcal{T}}}) \star P^+ = Q^+/{\approx}$ of semigroups. In other words, two state sequences $\bm{p}, \bm{q} \in Q^+$ are equivalent under $\approx$ if they are the same up to reductions within the $S^+$ blocks.
    To simplify our notation, we extend $\approx$ into a relation on $Q^*$ by letting $\varepsilon \approx \varepsilon$ (and $\varepsilon \not\approx \bm{p}$ for all $\bm{p} \in Q^+$).
    In this free product, we may define the \emph{normal form} of some word $\bm{p}_0 \bm{t}_1 \bm{p}_1 \ldots \bm{t}_n \bm{p}_n$ with $\bm{p}_0, \bm{p}_n \in P^*$, $\bm{p}_1, \dots, \bm{p}_{n - 1} \in P^+$ and $\bm{t}_1, \dots, \bm{t}_n \in S^+$ as $\bm{p}_0 s_1 \bm{p}_1 \dots s_n \bm{p}_n$ where $s_i$ is the unique element $s_i \in S$ with $s_i =_{\mathcal{T}} \bm{t}_i$ for $1 \leq i \leq n$.
    \begin{example}
      If our automaton $\mathcal{T}$ contains an identity state $e$ and we choose $S = \{ e \}$, then taking the normal form of some $\bm{p} \in Q^+$ is to reduce all blocks of the form $e^n$ to a single $e$. Accordingly, $\bm{p} \approx \bm{q}$ holds if and only if we end up with the same word when doing this for $\bm{p}$ and for $\bm{q}$.
    \end{example}
    
    It is not difficult to see that $Q^+/{=_{\mathcal{T}}}$, the semigroup generated by $\mathcal{T}$, is a quotient of $Q^+/{\approx}$:
    \begin{fact}\label{fct:approxImpliesEquiv}
      We have $\bm{p} \approx \bm{q} \implies \bm{p} =_{\mathcal{T}} \bm{q}$ for all $\bm{p}, \bm{q} \in Q^+$.
    \end{fact}
    
    Similarly, the action $\bm{q} \cdot u$ is compatible with the classes of $\approx$ since it is compatible with the classes of $=_{\mathcal{T}}$ (i.\,e.\ the semigroup generate by $\mathcal{T}$):
    \begin{fact}\label{fct:approxCdot}
      We have $\bm{p} \approx \bm{q} \implies \bm{p} \cdot u \approx \bm{q} \cdot u$ for all $\bm{p}, \bm{q} \in Q^+$ and $u \in \Sigma^*$.
    \end{fact}
    
    To further simplify our notation from here on, we fix some language $R \subseteq Q^*$ and simply say that a finite word is \emph{expandable} if it is $R$-expandable. Similarly, we say that $x \in \Sigma^+$ \emph{expands} some $w \in \Sigma^*$ if it $R$-expands it.
    
    We now use the above free product to define a relation that will allow us to detect whether a word $w \in \Sigma^*$ is expandable or not. The idea is basically that two state sequences $\bm{p}$ and $\bm{q}$ are equivalent if there are runs $\lrun{v}{}{\bm{p}}{w}$ and $\lrun{v}{}{\bm{q}}{w}$ to the same word in the orbital transducer of $w$. However, we will do this only \enquote{up to equality in the $S^+$ blocks} and this is where the free product comes into play. Additionally, we also have to keep track of the classes of the inputs (with respect to $R$).
    \begin{definition}
      For $w \in \Sigma^*$, we define the \emph{expansion relation} $\mathcal{E}_{w} \subseteq Q^+ \times Q^+$ by
      \begin{align*}
        \bm{p}_1 \mathrel{\mathcal{E}_{w}} \bm{p}_2 \iff \exists \bm{p}_1', \bm{p}_2' \in Q^+ \text{ with }& \bm{p}_1 \approx \bm{p}_1', \bm{p}_2 \approx \bm{p}_2' \text{ and}\\
        &\text{runs } \lrun{v}{\bm{r}'_1}{\bm{p}_1'}{w} \text{ and } \lrun{v}{\bm{r}'_2}{\bm{p}_2'}{w} \text{ in } \mathcal{T} \circ w\\
        &\text{for some } v \in Q^* \circ w \text{ and } \bm{r}'_1, \bm{r}'_2 \in R \text{.}
      \end{align*}
    \end{definition}
    \noindent{}Note that $\mathcal{E}_w$ implicitly depends on $R$.
    
    \begin{example}[Expansion Relation of the Adding Machine]\label{ex:expansionAddingMachine}
      Let $R = Q^*$ and consider the adding machine from \autoref{ex:addingMachine} (for $S = \{ e \}$).
      We have, for example, $qe \mathrel{\mathcal{E}_{00}} qeqe$ since we have the runs
      \begin{center}
        \begin{tikzpicture}[auto, swap, baseline=(00.base)]
          \node (00) {$00$};
          \node[left=of 00] (11) {$11$};
          \node[left=of 11] (002) {$00$};
          
          \path[->]
            (00) edge[path] node {$\scriptsize/e^3$} (11)
            (11) edge node {$\scriptsize/q$} (002)
          ;
        \end{tikzpicture}
      \end{center}
      and
      \begin{center}
          \begin{tikzpicture}[auto, swap, baseline=(00.base)]
            \node (00) {$00$};
            \node[left=of 00] (11) {$11$};
            \node[left=of 11] (002) {$00$};
            \node[left=of 002] (112) {$11$};
            \node[left=of 112] (003) {$00$};
            
            \path[->]
              (00) edge[path] node {$\scriptsize/e^3$} (11)
              (11) edge node {$\scriptsize/q$} (002)
              (002) edge[path] node {$\scriptsize/e^3$} (112)
              (112) edge node {$\scriptsize/q$} (003)
              
            ;
          \end{tikzpicture}
      \end{center}
      in $\mathcal{T} \circ 00$ (see \autoref{fig:orbTransAddingMachine}).
      In fact, this generalizes and we have $qe \mathrel{\mathcal{E}_{0^\ell}} qeqe$ for all $\ell > 0$.
      It turns out that we have $\bm{p}_1 \mathrel{\mathcal{E}_{0^\ell}} \bm{p}_2$ whenever $\bm{p}_1, \bm{p}_2 \in \left( e^* q e^+ \right)^*$ (see, later, in \autoref{ex:NFRAaddingMachine}).
    \end{example}
    
    \begin{example}[Expansion Relation for the Running Example]
      Let us also set $R = Q^*$ for our running example from \autoref{ex:runningAut} (with $S = \{ e \}$).
      From \autoref{fig:orbTransRunning}, we obtain the runs
      \begin{center}
        \begin{tikzpicture}[auto, swap, baseline=(00.base)]
          \node (22) {$22$};
          \node[left=of 22] (12) {$12$};
          \node[left=of 12] (02) {$02$};
          
          \path[->]
            (22) edge node {$\scriptsize p/e$} (12)
            (12) edge node {$\scriptsize p/r$} (02)
          ;
        \end{tikzpicture}
      \end{center}
      and
      \begin{center}
        \begin{tikzpicture}[auto, swap, baseline=(00.base)]
          \node (22) {$22$};
          \node[left=of 22] (21) {$21$};
          \node[left=of 21] (11) {$11$};
          \node[left=of 11] (02) {$02$};
          
          \path[->]
            (22) edge node {$\scriptsize q/e$} (21)
            (21) edge node {$\scriptsize p/e$} (11)
            (11) edge node {$\scriptsize p/p$} (02)
          ;
        \end{tikzpicture}
      \end{center}
      in $\mathcal{T} \circ 22$.
      This shows $re \mathrel{\mathcal{E}_{22}} pe \approx pee$.
      Similarly, one can also, for example, observe $ere \mathrel{\mathcal{E}_{22}} e$.
    \end{example}
    
    We will show that $\mathcal{E}_w$ can be used to test whether $w$ is expandable:
    \begin{proposition}\label{prop:EwAndExpandability}
      A word $x \in \Sigma^*$ expands a word $w \in \Sigma^*$ if and only if there are some $\bm{p}_1, \bm{p}_2 \in Q^+$ with $\bm{p}_1 \mathrel{\mathcal{E}_w} \bm{p}_2$ and $\bm{p}_1 \circ x \neq \bm{p}_2 \circ x$.
    \end{proposition}
    \begin{proof}
      Consider the surjective map $\pi: R \circ wx \to R \circ w$ given by $\bm{r} \circ wx \mapsto \bm{r} \circ w$. It is injective if and only if $|R \circ wx| = |R \circ w|$ (i.\,e.\ $x$ does not expand $w$).
      
      For the first direction, suppose that it is not injective. Then, there are $\bm{r}_1, \bm{r}_2 \in R$ with $\bm{r}_1 \circ wx \neq \bm{r}_2 \circ wx$ but $\bm{r}_1 \circ w = v = \bm{r}_2 \circ w$. The latter implies that there are runs $\lrun{v}{\bm{r}_1}{\bm{p}_1}{w}$ and $\lrun{v}{\bm{r}_2}{\bm{p}_2}{w}$ in $\mathcal{T} \circ w$ with $\bm{p}_1 = \bm{r}_1 \cdot w$ and $\bm{p}_2 = \bm{r}_2 \cdot w$, which means that we have $\bm{p}_1 \mathrel{\mathcal{E}_w} \bm{p}_2$. The former implies $v (\bm{p}_1 \circ x) = (\bm{r}_1 \circ w) (\bm{r}_1 \cdot w \circ x) = \bm{r}_1 \circ wx \neq \bm{r}_2 \circ wx = (\bm{r}_2 \circ w) (\bm{r}_2 \cdot w \circ x) = v (\bm{p}_2 \circ x)$. Thus, we have $\bm{p}_1 \circ x \neq \bm{p}_2 \circ x$.
      
      For the other direction, suppose that $\pi$ is injective and that we have $\bm{p}_1, \bm{p}_2 \in Q^+$ with $\bm{p}_1 \mathrel{\mathcal{E}_w} \bm{p}_2$. By definition, there are $\bm{p}_1', \bm{p}_2'$ with $\bm{p}_1 \approx \bm{p}_1'$, $\bm{p}_2 \approx \bm{p}_2'$ and runs $\lrun{v}{\bm{r}_1'}{\bm{p}_1'}{w}$ and $\lrun{v}{\bm{r}_2'}{\bm{p}_2'}{w}$ in $\mathcal{T} \circ w$. Thus, we have $\bm{r}_1' \circ w = v = \bm{r}_2' \circ w$ and, since $\pi$ is injective, also $\bm{r}_1' \circ wx = \bm{r}_2' \circ wx$. We obtain $v (\bm{p}_1' \circ x) = (\bm{r}_1' \circ w) (\bm{r}_1' \cdot w \circ x) = \bm{r}_1' \circ wx = \bm{r}_2' \circ wx = (\bm{r}_2' \circ w) (\bm{r}_2' \cdot w \circ x) = v (\bm{p}_2' \circ x)$, which implies $\bm{p}_1 \circ x = \bm{p}_1' \circ x = \bm{p}_2' \circ x = \bm{p}_2 \circ x$ since we have $\bm{p}_1 =_{\mathcal{T}} \bm{p}_1'$ and $\bm{p}_2 =_{\mathcal{T}} \bm{p}_2'$ by \autoref{fct:approxImpliesEquiv}.
    \end{proof}
    
    An immediate consequence of the last result is that expandability only depends on the expansion relation:
    \begin{proposition}\label{prop:OrbitIncreaseDependsOnEw}
      For all $u, v, x \in \Sigma^*$, we have
      \[
        \mathcal{E}_u = \mathcal{E}_v \implies \left( |R \circ u| < |R \circ ux| \iff |R \circ v| < |R \circ vx| \right) \text{.}
      \]
    \end{proposition}
    \begin{proof}
      Let $\mathcal{E}_u = \mathcal{E}_v$ and $|R \circ u| < |R \circ ux|$. By \autoref{prop:EwAndExpandability}, there are $\bm{p}_1, \bm{p}_2 \in Q^+$ with $\bm{p}_1 \mathrel{\mathcal{E}_u} \bm{p}_2$ and $\bm{p}_1 \circ x \neq \bm{p}_2 \circ x$. The former is the same as $\bm{p}_1 \mathrel{\mathcal{E}_v} \bm{p}_2$ and, from the other direction of \autoref{prop:EwAndExpandability}, we obtain $|R \circ v| < |R \circ vx|$.
    \end{proof}

    We have just seen that it suffices to know the expansion relation for a word $w$ to answer the question whether a given $x$ expands $w$. Next, we will see that the expansion relation is compatible with appending a word.
    \begin{proposition}\label{prop:EwCompatible}
      For all $u, v, x \in \Sigma^*$, we have $\mathcal{E}_u = \mathcal{E}_v \implies \mathcal{E}_{ux} = \mathcal{E}_{vx}$.
    \end{proposition}
    \begin{proof}
      We will only show the statement for $x = a \in \Sigma$ since the general case follows from this using an induction. Assume that we have $\mathcal{E}_u = \mathcal{E}_v$ for some $u, v \in \Sigma^*$. If we have $\bm{q}_1 \mathrel{\mathcal{E}_{ua}} \bm{q}_2$ for some $\bm{q}_1, \bm{q}_2 \in Q^+$, then there are $\bm{q}_1', \bm{q}_2' \in Q^+$ and $\bm{r}'_1, \bm{r}'_2 \in R$ with $\bm{q}_1 \approx \bm{q}_1'$, $\bm{q}_2 \approx \bm{q}_2'$ and runs $\lrun{u'a'}{\bm{r}'_1}{\bm{q}_1'}{ua}$ and $\lrun{u'a'}{\bm{r}'_2}{\bm{q}_2'}{ua}$ in $\mathcal{T} \circ ua$ for some $u' \in \Sigma^*$ and $a' \in \Sigma$. By the construction of $\mathcal{T} \circ ua$ this is only possible if there are $\bm{p}_1', \bm{p}_2' \in Q^+$ with $\bm{p}_1' \cdot a = \bm{q}_1'$, $\bm{p}_2' \cdot a = \bm{q}_2'$, $\bm{p}_1' \circ a = a' = \bm{p}_2' \circ a$ and runs $\lrun{u'}{\bm{r}'_1}{\bm{p}_1'}{u}$ and $\lrun{u'}{\bm{r}'_2}{\bm{p}_2'}{u}$ in $\mathcal{T} \circ u$, which implies that we have $\bm{p}'_1 \mathrel{\mathcal{E}_u} \bm{p}'_2$ and, thus, $\bm{p}'_1 \mathrel{\mathcal{E}_v} \bm{p}'_2$.

      This, in turn, implies that there are $\bm{p}_1'', \bm{p}_2'' \in Q^+$ and $\bm{r}''_1, \bm{r}''_2 \in R$ with $\bm{p}_1'' \approx \bm{p}_1'$, $\bm{p}_2'' \approx \bm{p}_2'$ and runs $\lrun{v'}{\bm{r}''_1}{\bm{p}''_1}{v}$ and $\lrun{v'}{\bm{r}''_2}{\bm{p}_2''}{v}$ in $\mathcal{T} \circ v$ for some $v' \in \Sigma^*$. From this follows $\bm{p}_1' =_{\mathcal{T}} \bm{p}_1''$ and $\bm{p}_2' =_{\mathcal{T}} \bm{p}_2''$ by \autoref{fct:approxImpliesEquiv}, which, in particular, means $a' = \bm{p}_1' \circ a = \bm{p}_1'' \circ a$ and $a' = \bm{p}_2' \circ a = \bm{p}_2'' \circ a$. Thus, we have runs $\lrun{v'a'}{\bm{r}_1''}{\bm{q}_1''}{va}$ and $\lrun{v'a'}{\bm{r}_2''}{\bm{q}_2''}{va}$ in $\mathcal{T} \circ va$ for $\bm{q}_1'' = \bm{p}_1'' \cdot a$ and $\bm{q}_2'' = \bm{p}_2'' \cdot a$. Finally, \autoref{fct:approxCdot} yields $\bm{q}_1'' = \bm{p}_1'' \cdot a \approx \bm{p}_1' \cdot a = \bm{q}_1' \approx \bm{q}_1$ and $\bm{q}_2'' = \bm{p}_2'' \cdot a \approx \bm{p}_2' \cdot a = \bm{q}_2' \approx \bm{q}_2$, which shows $\bm{q}_1 \mathrel{\mathcal{E}_{va}} \bm{q}_2$. The other direction is symmetric.
    \end{proof}
  \end{section}
  \begin{section}{Bounded Activity}
    In this section, we will use the previous results to show that the $\omega$-words with infinite $R$-orbits form a deterministic $\omega$-regular language if the automaton $\mathcal{T}$ is of bounded $S$-activity and $R$ is regular. Therefore, we assume the latter from now on and that there is some constant $K$ bounding all $\alpha(n)$ with $n \in \mathbb{N}$.

    \begin{subsection}{A Finite Acceptor for the Expansion Relation}\label{ssec:NFRAs}
      The basic idea of our approach is that we will store the expansion relation using a finite acceptor whose size is uniformly bounded. While we could use nondeterministic finite-state transducers for this, we will instead introduce a different acceptor model that more closely mimics a minimization of the orbital transducer or, in fact, the product $\mathcal{T} \circ w \times R$.
      
      \begin{definition}
        A nondeterministic finite relation acceptor (NFRA) $\mathcal{A}$ is a tuple $(Z, \Gamma, \tau,\allowbreak z_0, \mathcal{F})$ where $Z$ is a set of \emph{states}, $\Gamma$ is an alphabet, $\tau \subseteq Z \times \Gamma \times Z$ is the set of \emph{transitions}, $z_0 \in Z$ is the \emph{initial state} and $\mathcal{F} \subseteq Z \times Z$ is the \emph{acceptance relation}. For a transition $(y, a, z)$, we also use the graphical notations $\transa{y}{a}{z}$ and $\ltransa{z}{a}{y}$.
        
        A \emph{run} of the NFRA $\mathcal{A} = (Z, \Gamma, \tau, z_0, \mathcal{F})$ on a word $a_1 \dots a_n$ with $a_1, \dots, a_n \in \Gamma$ is a sequence
        \vspace*{-\baselineskip}
        \begin{center}
          \begin{tikzpicture}[auto, swap, baseline=(y0.base)]
            \node (y0) {$y_0$};
            \node[left=of y0] (y1) {$y_1$};
            \node[left=of y1] (dots) {$\dots$};
            \node[left=of dots] (yn) {$y_n$};
            
            \path[->] (y0) edge node {$a_1$} (y1)
                      (y1) edge node {$a_2$} (dots)
                      (dots) edge node {$a_n$} (yn)
            ;
          \end{tikzpicture}
        \end{center}
        with $\transa{y_{i - 1}}{a_i}{y_i} \in \tau$ for all $1 \leq i \leq n$.
        It \emph{starts} in $y_0$ and \emph{ends} in $y_n$.
        It is \emph{initial} if $y_0 = z_0$.
        In short-hand notation, we also denote this run by $\lruna{y_n}{a_n \dots a_1}{y_0}$.
        
        An NFRA $\mathcal{A} = (Z, \Gamma, \tau, z_0, \mathcal{F})$ \emph{accepts} a pair $(u, v) \in \Gamma^* \times \Gamma^*$ if it admits an initial run ending in $z$ on $u$ and an initial run ending in $y$ on $v$ with $z \mathrel{\mathcal{F}} y$. The \emph{recognized relation} for $\mathcal{A}$ is $\mathscr{R}(\mathcal{A}) = \{ (u, v) \in \Gamma^* \times \Gamma^* \mid \mathcal{A} \text{ accepts } (u, v) \}$.
      \end{definition}
      \begin{remark}
        For readers familiar with rational subsets of monoids (see e.\,g.\ \cite{edam}), we want to point out that $\mathscr{R}(\mathcal{A})$ is indeed a rational relation (in the sense that it is a rational subset of the monoid $\Gamma^* \times \Gamma^*$). This may be seen from observing that any NFRA $\mathcal{A} = (Z, \Gamma, \tau, z_0, \mathcal{F})$ may be turned into a (more traditional) non-deterministic, finite-state (spelling) $\Gamma^* \times \Gamma^*$-acceptor (again, see e.\,g.\ \cite{edam} for more information on $M$-automata or, as we call them here, $M$-acceptors). The idea is to take $Z \times Z$ as the state set, $(z_0, z_0)$ as the initial state, $\mathcal{F}$ as the set of final states and add the transitions
        \begin{align*}
          &\{ \transa{(x, y)}{(a, \varepsilon)}{(x', y)} \mid \transa{x}{a}{x'} \in \tau, y \in Z \}\\
          {}\cup{}&
          \{ \transa{(x, y)}{(\varepsilon, a)}{(x, y')} \mid \transa{y}{a}{y'} \in \tau, x \in Z \} \text{.}
        \end{align*}
        It is now a routine matter to check that the NFRA $\mathcal{A}$ accepts a pair $(u, v)$ if and only if there is an initial run ending in a final state in the thus constructed acceptor where the concatenation of the first components of the transitions on that run is exactly $u$ and the concatenation of the second components is exactly $v$.
      \end{remark}
      
      We will describe an NFRA that recognizes the expansion relation.
      However, before we do this in full technical detail, we first describe the general idea.

      \paragraph*{Idea of the construction.}
        To explain the construction, it is helpful to first consider the special case where $R = Q^*$ and $S=\{ e \}$ for an identity state $e$.
        For this choice, there is only one class of $R$ and we de not need to keep track of this information.
        In order to simply things further, we also re-define $\bm{p} \approx \bm{q}$ (for $\bm{p}, \bm{q} \in Q^*$) for this special case for now: we let $\bm{p} \approx \bm{q}$ if $\bm{p}$ and $\bm{q}$ are the same word up to removing all occurrences of the identity state $e$.\footnote{The difference to the previous definition is that there a block of $e$ states was shortened into a single $e$ but not completely dropped.}
        
        Suppose we want to encode the expansion relation $\mathcal{E}_w$ for some $w \in \Sigma^*$.
        We will keep track of whether there is a path $\lrun{v}{}{\bm{q}}{w}$ in the orbital graph $\mathcal{T} \circ w$ for some $v \in \Sigma^*$ and $\bm{q} \in Q^*$.
        The main idea is that we may factorize $\bm{q} = \bm{r}_{\ell} q_\ell \cdots \bm{r}_1 q_1 \bm{r}_0$ for $\bm{r}_0, \dots, \bm{r}_\ell \in \{ e \}^*$ and $q_1, \dots, q_\ell \in Q$.
        Consider a run
        \begin{center}
          \begin{tikzpicture}[auto, swap, baseline=(y0.base)]
            \node (w0) {$w_0 = w$};
            \node[left=of w0] (v0) {$v_0$};
            \node[left=of v0] (w1) {$w_1$};
            \node[left=of w1] (dots) {$\dots$};
            \node[left=of dots] (wl) {$w_\ell$};
            \node[left=of wl] (vl) {$v_\ell$};
            
            \path[->] (w0) edge[path] node {$\scriptsize/\bm{r}_0$} (v0)
                      (v0) edge node {$\scriptsize/q_1$} (w1)
                      (w1) edge[path] node {$\scriptsize/\bm{r}_1$} (dots)
                      (dots) edge node {$\scriptsize/q_\ell$} (wl)
                      (wl) edge[path] node {$\scriptsize/\bm{r}_\ell$} (vl)
            ;
          \end{tikzpicture}
        \end{center}
        in $\mathcal{T} \circ w$. By definition, the word $w_i$ is active for the state $q_i$ (for $1 \leq i \leq \ell$) and the number of such active words is bounded by a constant $K$.
        The main idea is now that, since the $\bm{r}_i \in \{ e \}^*$ will not contribute to expanding the orbit, we may compress this run into
        \begin{center}
          \begin{tikzpicture}[auto, swap, baseline=(w0.base)]
            \node (w0) {$w_0 = w$};
            \node[left=of w0] (w1) {$w_1$};
            \node[left=of w1] (dots) {$\dots$};
            \node[left=of dots] (wl) {$w_\ell$};
            
            \path[->]
              (w0) edge node {$\scriptsize/q_1$} (w1)
              (w1) edge node {$\scriptsize/q_2$} (dots)
              (dots) edge node {$\scriptsize/q_\ell$} (wl)
            ;
          \end{tikzpicture}.
        \end{center}
        Formally, we use the active words $u$ of length $|w|$ and $w$ as the states of the NFRA and define
        \[
          \mathscr{W}(u) = \{ v \mid \lrun{v}{}{\bm{r}}{u} \text{ for some } \bm{r} \in \{ e \}^* \} \text{.}
        \]
        We then have a transition $\ltransa{u'}{q}{u}$ (for $q \in Q$) if there is some $v \in \mathscr{W}(u)$ with $\ltrans{u'}{}{q}{v}$.
        By construction, we now have a run $\lruna{v}{\bm{q}}{w}$ in this NFRA if and only if there is a run from $w$ to $v$ in $\mathcal{T} \circ w$ whose output is $\approx$-equivalent to $\bm{q}$.
        We can extend this run to any $w' \in \mathscr{W}(v)$ without changing the class of the output (with respect to $\approx$).
        Thus, we have $\bm{p} \mathrel\mathcal{E}_w \bm{q}$ if and only if there are runs $\lruna{v}{\bm{p}}{w}$ and $\lruna{v'}{\bm{q}}{w}$ and some $w' \in \mathscr{W}(v) \cap \mathscr{W}(v')$. We, therefore, use $v \mathrel\mathcal{F} v'$ if and only if $\mathscr{W}(v) \cap \mathscr{W}(v') \neq \emptyset$ as the acceptance relation of the NFRA.
        
        This way, we have defined an NFRA with at most $K$ states for the expansion relation $\mathcal{E}_w$.
        For more general $R$, we also have to keep track of the input but only up to its class (with respect to $R$) and, for more general $S$, we cannot simply remove all $e$ states but need to keep track of the image in $S$.
        This makes the actual construction more technical.
      
      \paragraph*{Back to the General Case.}
      We now return to the general case and proceed by describing the NFRA $\mathcal{A}_w$ that recognizes $\mathcal{E}_w$ for some word $w \in \Sigma^*$, which we fix for the rest of this subsection to simplify our notation.
      Additionally, we write $A = A(|w|)$ for the set of $S$-active words of length $|w|$.
      Recall that we have $|A| \leq K$ and write $P$ for $P = Q \setminus S$.
      \begin{definition}\label{def:Aw}
        Let
        \begin{align*}
          Z ={}&\{ (u, \varepsilon, C, C) \mid u \in A \cup \{ w \}, C \text{ class of } R \} \cup{}\\
          & \{ (u, s, C, D) \mid u \in A \cup \{ w \}, s \in S, C \text{ and } D \text{ classes of } R \} \text{.}
        \end{align*}
        For every $\bm{t} \in S^+$, $u \in A \cup \{ w \}$ and every pair $C, D$ of classes of $R$, there is exactly one element $(u, s, C, D)$ in $Z$ with $\bm{t} =_{\mathcal{T}} s$. For convenience, we write $(u, \bm{t}, C, D)$ for this unique element.
        
        Define\footnote{Recall that we use $\mathcal{P}(X)$ to denote the powerset of $X$.}
        \begin{align*}
          \mathscr{W}: Z &\to \mathcal{P}(Q^* \circ w) \\
          (u, \varepsilon, C, C) &\mapsto \{ u \}\\
          (u, s, C, D) &\mapsto \{ u' \mid \exists \bm{t} \in S^+: s =_{\mathcal{T}} \bm{t}, \lrun{(u', D)}{}{\bm{t}}{(u, C)} \text{ in } \mathcal{T} \circ w \times R \} \text{ for } s \in S \text{.}
        \end{align*}
        and let $\mathcal{A}_w$ denote the NFRA with states $Z$, alphabet $Q$, transitions
        \begin{align*}
          & \{ \transa{(u, \varepsilon, C, C)}{s}{(u, s, C, D)} \mid
            \begin{aligned}[t]
              &u \in A \cup \{ w \}, C, D \text{ classes of } R, s \in S \}
            \end{aligned}\\
          {}\cup{}& \{ \transa{(u, s, C, D)}{t}{(u, st, C, D)} \mid
            \begin{aligned}[t]
              &u \in A \cup \{ w \}, C, D \text{ classes of } R, s, t \in S \}
            \end{aligned}\\
          {}\cup{}& \{ \transa{(u, \varepsilon, C, C)}{p}{(v, \varepsilon, E, E)} \mid \begin{aligned}[t]
            &u \in A \cup \{ w \}, C \text{ class of } R,\\
            &\ltrans{(v, E)}{}{p}{(u, C)} \text{ in } \mathcal{T} \circ w \times R\\
            &\text{for } p \in P \text{ and some class $E$ of $R$} \} 
          \end{aligned}\\
          {}\cup{}& \{ \transa{(u, s, C, D)}{p}{(v, \varepsilon, E, E)} \mid \begin{aligned}[t]
            &u, v \in A \cup \{ w \}, C, D \text{ classes of } R, s \in S,\\
            &\exists u' \in \mathscr{W}(u, s, C, D) \text{ with }\\
            &\ltrans{(v, E)}{}{p}{(u', D)} \text{ in } \mathcal{T} \circ w \times R\\
            &\text{for } p \in P \text{ and some class $E$ of $R$} \} \text{,}
          \end{aligned}
        \end{align*}
        initial state $(w, \varepsilon, \mathscr{C}(\varepsilon), \mathscr{C}(\varepsilon))$ and the accepting relation $\mathcal{F} \subseteq Z \times Z$ given by
        \[
          (u, s, C, D) \mathrel{\mathcal{F}} (v, t, E, F) \iff D, F \subseteq R, \mathscr{W}(u, s, C, D) \cap \mathscr{W}(v, t, E, F) \neq \emptyset \text{.}
        \]
      \end{definition}
      
      \begin{remark}
        We may define this acceptor for any language $R \subseteq Q^+$ but it only remains finite if $R$ has finitely many classes (i.\,e.\ if $R$ is regular).
        The size of the state set $Z$ of $\mathcal{A}_w$ is then $(|A \cup \{ w \}|)(1 \cdot N + |S| \cdot N^2) \leq (K + 1)N(|S|N + 1)$ where $N$ is the number of classes of $R$. Thus, the size of $\mathcal{A}_w$ is not only finite but uniformly bounded.
        Furthermore, note that $\mathcal{A}_w$ can clearly be computed from $w$, $\mathcal{T}$, $R$ and $S$.
      \end{remark}
      
      \begin{example}[NFRA for the Expansion Relation of the Adding Machine]\label{ex:NFRAaddingMachine}
        The NFRA $\mathcal{A}_w$ we obtain for $w = 0^\ell$ (where $\ell > 0$) for the adding machine from \autoref{ex:addingMachine} (with $R = Q^*$ and $S = \{ e \}$) is very simple.
        We have $A = \{ 0^\ell \}$ and, since we chose $R = Q^*$, only a single class $C$ of $R$.
        It, therefore, makes sense to simply write $(u)$ for the state $(u, \varepsilon, C, C)$ and $(u, e)$ for the state $(u, e, C, C)$;
        in fact, we only have such states for $u = w = 0^\ell$.
        We may depict $\mathcal{A}_w$ in the usual way of depicting automata as:
        
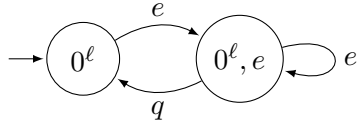
\begin{center}
          \begin{tikzpicture}[auto, shorten >=1pt, >=latex, initial text=]
            \node[state, initial] (0) {$0^\ell$};
            \node[state, right=of 0] (0e) {$0^\ell, e$};
            
            \draw[->]
              (0) edge[bend left] node {$e$} (0e)
              (0e) edge[bend left] node {$q$} (0)
              (0e) edge[loop right] node {$e$} (0e)
            ;
          \end{tikzpicture}
        \end{center}
        We have $\mathscr{W}(0^\ell) = \{ 0^\ell \}$ and $\mathscr{W}(0^\ell, e) = \{ 0, 1 \}^\ell$ (since we may reach any word $v \in \{ 0, 1 \}^\ell$ by a path $\lrun{v}{}{e^i}{0^\ell}$ for some $i$; compare to \autoref{fig:orbTransAddingMachine}).
        Since $\mathcal{W}(0^\ell) \cap \mathcal{W}(0^\ell, e) = \{ 0^\ell \} \neq \emptyset$, we have that all pairs of states $\{ 0^\ell, (0^\ell, e) \}^2$ are in $\mathcal{F}$.
        We will show in \autoref{prop:AwRecEw} that $\mathcal{A}_w$ recognizes $\mathcal{E}_w$ and, thus, obtain that $\bm{p}_1 \mathrel{\mathcal{E}_w} \bm{p}_2$ whenever $\bm{p}_1$ and $\bm{p}_2$ can both be read from $0^\ell$ in the above NFRA.
        This is the case if and only if $\bm{p}_1, \bm{p}_2 \in \left( e^* q e^+ \right)^*$ (as claimed in \autoref{ex:expansionAddingMachine}).
      \end{example}
      
      \begin{example}[NFRA for the Running Example]
        Recall our running example from \autoref{ex:runningAut} and (continue to) let $S = \{ e \}$, $R = Q^*$ and $w = 22$.
        Since we again only have a single class $C$ of $R$, we continue to simply write $u$ for the state $(u, \varepsilon, C, C)$ and $(u, e)$ for $(u, e, C, C)$.
        The active words are already marked in the orbital transducer $\mathcal{T} \circ 22$ from \autoref{fig:orbTransRunning}.
        We have $\{ w \} \cup A = \{ 02, 20, 21, 22 \}$ and, thus,
        \[
          Z = \{ 02, (02, e), 20, (20, e), 21, (21, e), 22, (22, e) \} \text{.}
        \]
        By definition, we have $\mathscr{W}(u) = \{ u \}$ for $u \in \{ 02, 20, 21, 22 \}$.
        For the states of the form $(u, e)$, $\mathscr{W}(u, e)$ contains those words $v$ that are reachable from $u$ by a run $\lrun{v}{}{e^i}{u}$ for some $i > 0$ (but observe that we always have an $e/e$-self-loop at every state in $\mathcal{T} \circ 22$).
        This yields:
        \begin{align*}
          \mathscr{W}(02, e) &= \{ 02, 11, 12, 21 \}, &
          \mathscr{W}(20, e) &= \{ 10, 11, 20, 21 \}, \\
          \mathscr{W}(21, e) &= \{ 11, 21 \} \qquad \text{and} &
          \mathscr{W}(22, e) &= \{ 11, 12, 21, 22 \}
        \end{align*}
        Observe that $11$ is contained in all sets $\mathscr{W}(u, e)$.
        Thus, we have $(u, e) \mathrel{\mathcal{F}} (v, e)$ for all $u, v \in \{ w \} \cup A$.
        We also always have $\mathscr{W}(u) \cap \mathscr{W}(u, e) = \{ u \} \neq \emptyset$ (because of the $e/e$-self-loops) and, thus, $u \mathrel{\mathcal{F}} (u, e)$ for all $u \in \{ w \} \cup A$.
        We also have $21 \in \mathscr{W}(u, e)$ for all such $u$ and, thus, $21, (21, e) \mathrel{\mathcal{F}} (u, e)$.
        
        For the transitions, we obtain $\transa{u}{e}{(u, e)}$ and $\transa{(u, e)}{e}{(u, e)}$ for all $u \in \{ w \} \cup A$ from the first and second line of the definition.
        Observe that no state $u \in \{ w \} \cup A$ has an out-going transition with an output different to $e$ in $\mathcal{T} \circ 22$ (although such transitions may exist in general!) and that, thus, we do not create any transition using the third line of the definition.
        The last line does create transitions, however.
        For example, consider the transition $\trans{12}{p}{r}{02}$ in $\mathcal{T} \circ 22$.
        We have $12 \in \mathscr{W}(02, e), \mathscr{W}(22, e)$ and, therefore, obtain the transitions $\transa{(02, e)}{r}{02}$ and $\transa{(22, e)}{r}{02}$ in $\mathcal{A}_w$.
        
        In the same way, all the (reachable) transitions with output different to $e$ create transitions in $\mathcal{A}_w$ and, in the end, we obtain:
        
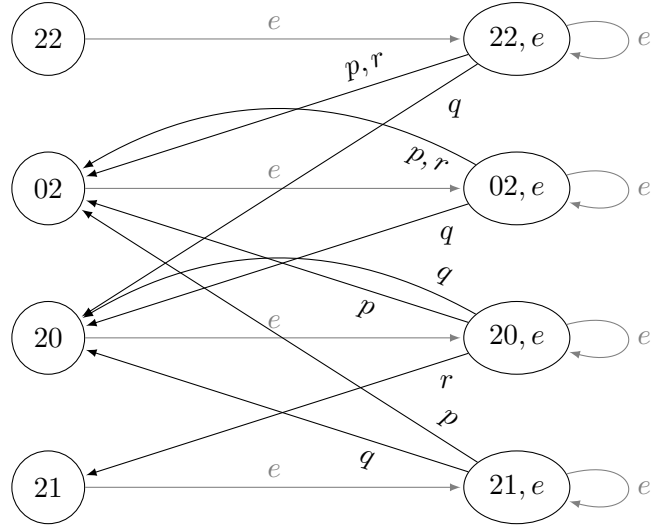
\begin{center}
          \begin{tikzpicture}[auto, shorten >=1pt, >=latex, node distance=1cm and 5cm]
            \node[state] (22) {$22$};
            \node[state, right=of 22, ellipse] (22e) {$22, e$};
            \node[state, below=of 22] (02) {$02$};
            \node[state, right=of 02, ellipse] (02e) {$02, e$};
            \node[state, below=of 02] (20) {$20$};
            \node[state, right=of 20, ellipse] (20e) {$20, e$};
            \node[state, below=of 20] (21) {$21$};
            \node[state, right=of 21, ellipse] (21e) {$21, e$};
            
            \draw[->, gray]
              (22) edge node {$e$} (22e)
              (22e) edge[loop right] node {$e$} (22e)
              (02) edge node {$e$} (02e)
              (02e) edge[loop right] node {$e$} (02e)
              (20) edge node {$e$} (20e)
              (20e) edge[loop right] node {$e$} (20e)
              (21) edge node {$e$} (21e)
              (21e) edge[loop right] node {$e$} (21e)
            ;
            
            \draw[->, pos=0.1]
              (22e) edge node[sloped, above, pos=0.25] {$p, r$} (02)
                    edge node {$q$} (20)
              (02e) edge[bend right] node[sloped, below] {$p, r$} (02)
                    edge node {$q$} (20)
              (20e) edge node[sloped, below, pos=0.25] {$p$} (02)
                    edge[bend right] node[sloped, above] {$q$} (20)
                    edge node {$r$} (21)
              (21e) edge node[sloped, above] {$p$} (02)
                    edge node[sloped, below, pos=0.25] {$q$} (20)
            ;
          \end{tikzpicture}
        \end{center}
        Note that, while this automaton is deterministic, this is not the case in general!
      \end{example}

      The states of the NFRA represent parts of the orbital transducer $\mathcal{T} \circ w$ or -- more precisely -- the product $\mathcal{T} \circ w \times R$. This correspondence is given by $\mathscr{W}$: states of the form $(u, \varepsilon, C, C)$ represent the (constantly many) words from $A \cup \{ w \}$ with an in-going transition whose output is not from $S$ and the other states of the form $(u, s, C, D)$ represent the words reachable with output equal to $s$ in $\mathscr{S}(\mathcal{T})$ from those words in $A \cup \{ w \}$. We have basically attached a copy of the (finite) Cayley graph of $S^+ / {=_{\mathcal{T}}}$ to every $u \in A \cup \{ w \}$. Additionally, we continue to have those transitions whose output is not from $S$.
      
      While we are mostly interested in the output of the transitions in $\mathcal{T} \circ w$, we use the last two components of a state to track the class of $R$ for the input. The first one is used to store the class of the input up to when we entered the last word from $A \cup \{ w \}$ and, in the second one, we nondeterministically guess the class of the input when we leave the local copy of the Cayley graph. Of course, it is possible that the nondeterministic choice is impossible to realize; in this case, we have that $\mathscr{W}$ of a state is the empty set.
      
      \begin{proposition}\label{prop:AwRecEw}
        The recognized relation of $\mathcal{A}_w$ is $\mathcal{E}_w$.
      \end{proposition}
      \begin{proof}
        First, we show the following (which we will refer to as Claim I): if we have a run $\lrun{v}{\bm{r}}{\bm{q}}{w}$ for some $\bm{r}, \bm{q} \in Q^*$ and $v \in \Sigma^*$ in $\mathcal{T} \circ w$, then, for all $\bm{q}' \in Q^*$ with $\bm{q}' \approx \bm{q}$, we have a run
        \[
          \lruna{(u, s, C, \mathscr{C}(\bm{r}))}{\bm{q}'}{ (w, \varepsilon, \mathscr{C}(\varepsilon), \mathscr{C}(\varepsilon)) }
        \]
        in $\mathcal{A}_w$ for some $u \in A \cup \{ w \}$, $s \in S \uplus \{ \varepsilon \}$ and class $C$ of $R$ with $v \in \mathscr{W}(u, s, C, \mathscr{C}(\bm{r}))$. From the claim follows that $\bm{p} \mathrel{\mathcal{E}}_w \bm{q}$ implies that $\mathcal{A}_w$ accepts $(\bm{p}, \bm{q})$ for all $\bm{p}, \bm{q} \in Q^+$.
        
        In order to show Claim I, we first observe that we have a run $\lrun{v}{\bm{r}}{\bm{q}}{w}$ in $\mathcal{T} \circ w$ if and only if we have a run $\lrun{(v, \mathscr{C}(\bm{r}))}{\bm{r}}{\bm{q}}{(w, \mathscr{C}(\varepsilon))}$ in $\mathcal{T} \circ w \times R$. We use this fact and an induction on the letters from $P$ and the $S^+$ blocks of $\bm{q}$ to show the claim. For the empty state sequence, there is nothing to show as we have $\mathscr{W}(w, \varepsilon, \mathscr{C}(\varepsilon), \mathscr{C}(\varepsilon)) = \{ w \}$.
        
        Next, we consider the state sequences $\bm{s} \bm{q}$ with $\bm{s} \in S^+$ and $\bm{q} \in \{ \varepsilon \} \cup PQ^*$ and suppose that there is a run
        \begin{tikzpicture}[auto, shorten >=1pt, >=latex, baseline=(v.base), inner sep=0pt, outer xsep=0.3333em]
          \node (v) {$(u', D)$};%
          \setlength{\edgelength}{\widthof{\scriptsize$/\bm{s}$}+0.5cm}%
          \node[base right=\edgelength of v] (u) {$(u, C)$};%
          \path[<-] ([yshift=-0.25mm] v.mid east) edge[path] node[inner sep=0pt] {\scriptsize$/\bm{s}$} ([yshift=-0.25mm] u.mid west);%
          
          \setlength{\edgelength}{\widthof{\scriptsize$/\bm{q}$}+0.5cm}%
          \node[base right=\edgelength of u] (w) {$(w, \mathscr{C}(\varepsilon))$};%
          \path[<-] ([yshift=-0.25mm] u.mid east) edge[path] node[inner sep=0pt] {\scriptsize$/\bm{q}$} ([yshift=-0.25mm] w.mid west);%
        \end{tikzpicture}
        in $\mathcal{T} \circ w \times R$ for some $u, u' \in \Sigma^*$. For every $\bm{q}' \in Q^*$ with $\bm{q}' \approx \bm{q}$, we have a run $\lruna{(u, \varepsilon, C, C)}{\bm{q}'}{(w, \varepsilon, \mathscr{C}(\varepsilon), \mathscr{C}(\varepsilon))}$ in $\mathcal{A}_w$ by induction (or because $\bm{q}$ is empty). Note here that only states of the form $(u, \varepsilon, C, C)$ have ingoing transitions labeled with elements from $P$ (and that the initial state is of this form as well). To conclude this case we show that we have a run $\lruna{(u, \bm{s}, C, D)}{\bm{t}}{(u, \varepsilon, C, C)}$ for all $\bm{t} \in S^+$ with $\bm{t} \approx \bm{s}$. This is sufficient because we have $u' \in \mathscr{W}(u, \bm{s}, C, D)$ due to the run $\lrun{(u', D)}{}{\bm{s}}{(u, C)}$ in $\mathcal{T} \circ w \times R$. We write $\bm{t} = t_\ell \dots t_2 t_1$ for $t_1, \dots, t_\ell \in S$ and observe that we have a run
        \begin{center}
          \begin{tikzpicture}[auto, shorten >=1pt, >=latex, baseline=(v.base), inner sep=0pt, outer xsep=0.3333em]
            \node (u) {$(u, \varepsilon, C, C)$};%
            \setlength{\edgelength}{\widthof{\scriptsize$t_1$}+0.5cm}%
            \node[base left=\edgelength of u] (t1) {$(u, t_1, C, D)$};%
            \path[<-] (t1.mid east) edge node[inner sep=0pt] {\scriptsize$t_1$} (u.mid west);%
            
            \setlength{\edgelength}{\widthof{\scriptsize$t_2$}+0.5cm}%
            \node[base left=\edgelength of t1] (t2t1) {$(u, t_2 t_1, C, D)$};%
            \path[<-] (t2t1.mid east) edge node[inner sep=0pt] {\scriptsize$t_2$} (t1.mid west);%
            
            \setlength{\edgelength}{\widthof{\scriptsize$t_3$}+0.5cm}%
            \node[base left=\edgelength of t2t1] (dots) {$\dots$};%
            \path[<-] (dots.mid east) edge node[inner sep=0pt] {\scriptsize$t_3$} (t2t1.mid west);%
            
            \setlength{\edgelength}{\widthof{\scriptsize$t_\ell$}+0.5cm}%
            \node[base left=\edgelength of dots] (t) {$(u, \bm{t}, C, D)$};%
            \path[<-] (t.mid east) edge node[inner sep=0pt] {\scriptsize$t_\ell$} (dots.mid west);%
          \end{tikzpicture}
        \end{center}
        in $\mathcal{A}_w$ by construction.
        
        For the remaining inductive case, consider a state sequence $p \bm{q}$ with $p \in P$ and $\bm{q} \in Q^*$ such that there is a run
        \begin{tikzpicture}[auto, shorten >=1pt, >=latex, baseline=(v.base), inner sep=0pt, outer xsep=0.3333em]
          \node (v) {$(v, D)$};%
          \setlength{\edgelength}{\widthof{\scriptsize$/p$}+0.5cm}%
          \node[base right=\edgelength of v] (u) {$(u', C)$};%
          \path[<-] ([yshift=-0.25mm] v.mid east) edge node[inner sep=0pt] {\scriptsize$/p$} ([yshift=-0.25mm] u.mid west);%
          
          \setlength{\edgelength}{\widthof{\scriptsize$/\bm{q}$}+0.5cm}%
          \node[base right=\edgelength of u] (w) {$(w, \mathscr{C}(\varepsilon))$};%
          \path[<-] ([yshift=-0.25mm] u.mid east) edge[path] node[inner sep=0pt] {\scriptsize$/\bm{q}$} ([yshift=-0.25mm] w.mid west);%
        \end{tikzpicture}
        in $\mathcal{T} \circ w \times R$. Let $\bm{q}' \in Q^*$ be arbitrary with $\bm{q}' \approx \bm{q}$. By induction (or because $\bm{q}$ is empty), we have a run $\lruna{(u, s, E, C)}{\bm{q}'}{(w, \varepsilon, \mathscr{C}(\varepsilon), \mathscr{C}(\varepsilon))}$ in $\mathcal{A}_w$ for some $u \in A \cup \{ w \}$, $s \in S \uplus \{ \varepsilon \}$ and class $E$ of $R$ with $u' \in \mathscr{W}(u, s, E, C)$. We also have a transition $\ltransa{(v, \varepsilon, D, D)}{p}{(u, s, E, C)}$ in $\mathcal{A}_w$ by construction because we have the transition $\ltrans{(v, D)}{}{p}{(u', C)}$ in $\mathcal{T} \circ w \times R$. This concludes our proof of Claim I and the first direction as we have $\mathscr{W}(v, \varepsilon, D, D) = \{ v \}$.
        
        For the other direction, we show the similar Claim II: if we have a run 
        \[
          \lruna{(u, s, C, D)}{\bm{q}}{(w, \varepsilon, \mathscr{C}(\varepsilon), \mathscr{C}(\varepsilon))}
        \]
        in $\mathcal{A}_w$ for some $\bm{q} \in Q^*$, $u \in A \cup \{ w \}$, $s \in S \uplus \{ \varepsilon \}$ and classes $C, D$ of $R$, then, for every $u' \in \mathscr{W}(u, s, C, D)$, there is some $\bm{q}' \in Q^*$ with $\bm{q}' \approx \bm{q}$, some $\bm{r}' \in D$ and a run $\lrun{u'}{\bm{r}'}{\bm{q}'}{w}$ in $\mathcal{T} \circ w$. The claim implies that $\bm{p} \mathrel{\mathcal{E}_w} \bm{q}$ holds if $\mathcal{A}_w$ accepts some pair $(\bm{p}, \bm{q}) \in Q^+ \times Q^+$.
        
        Recall that we have a run $\lrun{u'}{\bm{r}'}{\bm{q}'}{w}$ in $\mathcal{T} \circ w$ if and only if we have the run $\lrun{(u', \mathscr{C}(\bm{r}'))}{\bm{r}'}{\bm{q}'}{(w, \mathscr{C}(\varepsilon))}$ in $\mathcal{T} \circ w \times R$. We use this and another induction on the letters from $P$ and the $S^+$ blocks of $\bm{q}$ to show Claim II. For the empty state sequence, there is nothing to show as we have $\mathscr{W}(w, \varepsilon, \mathscr{C}(\varepsilon), \mathscr{C}(\varepsilon)) = \{ w \}$.
        
        Therefore, we first consider a state sequence $\bm{s} \bm{q}$ with $\bm{s} \in S^+$ and $\bm{q} \in \{ \varepsilon \} \cup PQ^*$ such that we have a run
        \begin{tikzpicture}[auto, shorten >=1pt, >=latex, baseline=(V.base), inner sep=0pt, outer xsep=0.3333em]
          \node (V) {$(u, \bm{s}, C, D)$};%
          \setlength{\edgelength}{\widthof{\scriptsize$\bm{s}$}+0.5cm}%
          \node[base right=\edgelength of V] (U) {$(u, \varepsilon, C, C)$};%
          \path[<-] ([yshift=-0.25mm] V.mid east) edge[path] node[inner sep=0pt] {\scriptsize$\bm{s}$} ([yshift=-0.25mm] U.mid west);%
          
          \setlength{\edgelength}{\widthof{\scriptsize$\bm{q}$}+0.5cm}%
          \node[base right=\edgelength of U] (w) {$(w, \varepsilon, \mathscr{C}(\varepsilon), \mathscr{C}(\varepsilon))$};%
          \path[<-] ([yshift=-0.25mm] U.mid east) edge[path] node[inner sep=0pt] {\scriptsize$\bm{q}$} ([yshift=-0.25mm] w.mid west);%
        \end{tikzpicture}
        in $\mathcal{A}_w$ for some $u \in A \cup \{ w \}$ and classes $C, D$ of $R$. Note that the states on the left have to be of this form due to the construction of $\mathcal{A}_w$: only states of the form $(u, \varepsilon, C, C)$ have ingoing transitions labeled with an element from $P$ (and the initial state is of this form) and, from them, we can only reach states of the form $(u, \bm{s}, C, D)$ by a run labeled with $\bm{s}$. Since we have $\mathscr{W}(u, \varepsilon, C, C) = \{ u \}$, there is some $\bm{q}' \in Q^*$ with $\bm{q}' \approx \bm{q}$ and a run $\lrun{(u, C)}{}{\bm{q}'}{(w, \mathscr{C}(\varepsilon))}$ in $\mathcal{T} \circ w \times R$ by induction. We are done with this inductive case if we show that, for every $u' \in \mathscr{W}(u, \bm{s}, C, D)$, there is a run $\lrun{(u', D)}{}{\bm{t}}{(u, C)}$ in $\mathcal{T} \circ w \times R$ for some $\bm{t} \in S^+$ with $\bm{t} \approx \bm{s}$. Such a run exists, however, by the definition of $\mathscr{W}(u, \bm{s}, C, D)$.
        
        It remains to consider the state sequences $p\bm{q}$ with $p \in P$ and $\bm{q} \in Q^*$ such that there is a run
        \begin{tikzpicture}[auto, shorten >=1pt, >=latex, baseline=(V.base), inner sep=0pt, outer xsep=0.3333em]
          \node (V) {$(v, \varepsilon, E, E)$};%
          \setlength{\edgelength}{\widthof{\scriptsize$p$}+0.5cm}%
          \node[base right=\edgelength of V] (U) {$(u, s, C, D)$};%
          \path[<-] ([yshift=-0.25mm] V.mid east) edge node[inner sep=0pt] {\scriptsize$p$} ([yshift=-0.25mm] U.mid west);%
          
          \setlength{\edgelength}{\widthof{\scriptsize$\bm{q}$}+0.5cm}%
          \node[base right=\edgelength of U] (w) {$(w, \varepsilon, \mathscr{C}(\varepsilon), \mathscr{C}(\varepsilon))$};%
          \path[<-] ([yshift=-0.25mm] U.mid east) edge[path] node[inner sep=0pt] {\scriptsize$\bm{q}$} ([yshift=-0.25mm] w.mid west);%
        \end{tikzpicture}
        in $\mathcal{A}_w$ for some $u, v \in A \cup \{ w \}$, $s \in S \uplus \{ \varepsilon \}$ and classes $C, D, E$ of $R$. Again, the state on the left has to be of this form because only such states have ingoing transitions with an element from $P$ as the label. Also by the construction of $\mathcal{A}_w$, we obtain that there is some $u' \in \mathscr{W}(u, s, C, D)$ with a run $\lrun{(v, E)}{}{p}{(u', D)}$ in $\mathcal{T} \circ w \times R$ from the transition $\ltransa{(v, \varepsilon, E, E)}{p}{(u, s, C, D)}$ in $\mathcal{A}_w$. By induction, we also obtain that there is some $\bm{q}' \in Q^*$ with $\bm{q}' \approx \bm{q}$ and a run $\lrun{(u', D)}{}{\bm{q}'}{(w, \mathscr{C}(\varepsilon))}$ in $\mathcal{T} \circ w \times R$, which concludes our proof because we have $\mathscr{W}(v, \varepsilon, E, E) = \{ v \}$.
      \end{proof}
    \end{subsection}
    \begin{subsection}{\texorpdfstring{$\omega$-Regular}{ω-Regular}}
      We have seen (in \autoref{prop:OrbitIncreaseDependsOnEw}) that the growth behavior of the orbits only depends on the expansion relation and that those can be encoded in NFRAs of uniformly bounded size (\autoref{ssec:NFRAs}).
      Together, this allows us to encode the orbit growth behavior in a finite acceptor. More precisely, we may characterize the infinite words with infinite orbits as a deterministic Büchi language.
      \begin{theorem}\label{thm:InfiniteOrbitIsRegular}
        For a regular language $R \subseteq Q^*$, the language
        \[
          \{ \alpha \in \Sigma^\omega \mid |R \circ \alpha| = \infty \}
        \]
        is recognized by a deterministic Büchi acceptor. Furthermore, this acceptor can be computed from $\mathcal{T}$, $R$ and $S$.
      \end{theorem}
      \begin{proof}
        Using \autoref{def:Aw} and \autoref{prop:AwRecEw}, we can compute an NFRA $\mathcal{A}_w$ recognizing $\mathcal{E}_w$ with at most $(K + 1)N(|S|N + 1)$ states where $N$ is the number of classes of $R$ from $\mathcal{T} \circ w$.
        In particular, the set $Y = \{ \mathcal{A}_w \mid w \in \Sigma^* \}$ is finite and we use it as the state set of the sought deterministic Büchi acceptor. To compute $Y$ and the transitions of the BA, we use an iterative process:
        
        We start with $\mathcal{A}_\varepsilon$ as the initial state.
        As long as a state $\mathcal{A}_w$ has no outgoing transitions, we compute $\mathcal{A}_{wa}$ from $\mathcal{T} \circ wa$ for all $a \in \Sigma$.
        For each $a \in \Sigma$, there are two possible situations:
        The NFRA $\mathcal{A}_{wa}$ can be isomorphic to an existing state $\mathcal{A}_{w'}$ or not (where isomorphism is to be understood as that of abstract edge-labeled graphs with a marked initial state and a binary (acceptance) relation over the nodes).
        If it is not isomorphic to some $\mathcal{A}_{w'}$ (among the states constructed so far), we add $\mathcal{A}_{wa}$ as a new state together with the transition $\transa{\mathcal{A}_w}{a}{\mathcal{A}_{wa}}$.
        If it is isomorphic to a pre-existing state $\mathcal{A}_{w'}$, we add the transition $\transa{\mathcal{A}_w}{a}{\mathcal{A}_{w'}}$.
        In both cases, we mark the new transition as accepting if we have $|R \circ w| < |R \circ wa|$.
        
        This process of adding new states and transitions has to terminate because there are only finitely many possible NFRAs $\mathcal{A}_w$ and the resulting BA is deterministic.
        
        We do not necessarily end up in $\mathcal{A}_w$ after reading $w \in \Sigma^*$ from the initial state $\mathcal{A}_\varepsilon$.
        However, we may show (using an induction on the length of $w$) that the NFRA $\mathcal{A}_{w'}$ that we reach by reading $w$ from $\mathcal{A}_\varepsilon$ recognizes the expansion relation $\mathcal{E}_w$:
        
        This is clear for $w = \varepsilon$.
        For the inductive step, let $w = ua$ with $a \in \Sigma$.
        Let $\mathcal{A}_{u'}$ be the state we reach after reading $u$ from $\mathcal{A}_\varepsilon$.
        By its definition, $\mathcal{A}_{u'}$ recognizes the relation $\mathcal{E}_{u'}$ but, by induction, it must also recognize $\mathcal{E}_u$.
        Thus, we must have $\mathcal{E}_{u'} = \mathcal{E}_{u}$ and, by \autoref{prop:EwCompatible}, also $\mathcal{E}_{u'a} = \mathcal{E}_{ua}$, which is recognized by $\mathcal{A}_{u'a}$.
        When adding the out-going transitions for the state $\mathcal{A}_{u'}$, we distinguished whether $\mathcal{A}_{u'a}$ was isomorphic to a pre-existing state or not.
        If it was not, we added $\mathcal{A}_{u'a}$ as a new state together with the transition $\transa{\mathcal{A}_{u'}}{a}{\mathcal{A}_{u'a}}$.
        Thus, in this case, the state we reach after reading $ua$ from $\mathcal{A}_\varepsilon$ is $\mathcal{A}_{u'a}$, which indeed recognizes $\mathcal{E}_{ua}$ (as stated above).
        If $\mathcal{A}_{u'a}$ was isomorphic to some pre-existing state $\mathcal{A}_{w'}$, we added the transition $\transa{\mathcal{A}_{u'}}{a}{\mathcal{A}_{w'}}$ instead.
        Since $\mathcal{A}_{w'}$ is isomorphic to $\mathcal{A}_{u'a}$, it recognizes the same relation and we indeed have $\mathcal{E}_{ua} = \mathcal{E}_{u'a} = \mathcal{E}_{w'}$.
        
        Recall that the transition $\transa{\mathcal{A}_{u'}}{a}{\mathcal{A}_{u'a}}$ or $\transa{\mathcal{A}_{u'}}{a}{\mathcal{A}_{w'}}$, respectively, was accepting if $|R \circ u'| < |R \circ u'a|$.
        Since we have $\mathcal{E}_{u'} = \mathcal{E}_u$ (as just shown), this is the case if and only if $|R \circ u| < |R \circ ua|$ by \autoref{prop:OrbitIncreaseDependsOnEw}.
        Thus, whenever we pass an accepting transition in the Büchi acceptor, the orbit of the word we are reading has indeed increased.
        This means that its recognized $\omega$-language consists exactly of the $\omega$-words with an infinite $R$-orbit.
      \end{proof}
      \begin{remark}
        Using a straight-forward guess and check approach, one can see that the equivalence problem for NFRAs
        \problem{
          NFRAs $\mathcal{A}_1 = (Z_1, \Gamma, \tau_1, z_{0, 1}, \mathcal{F}_1)$ and $\mathcal{A}_2 = (Z_2, \Gamma, \tau_2, z_{0, 2}, \mathcal{F}_2)$
        }{
          is $\mathscr{R}(\mathcal{A}_1) = \mathscr{R}(\mathcal{A}_2)$?
        }\noindent
        is certainly in $\ComplexityClass{PSpace}$ (see e.\,g.\ \cite{papadimitriou97computational} or \cite{waechter2024word} for more information on complexity theory).
        Therefore, we can replace the check whether an isomorphic NFRA already exists as a state by a check whether an equivalent NFRA already exists in the construction of the Büchi accepter above.
        This yields a potentially smaller acceptor.
      \end{remark}
      
      For a finite word $v \in \Sigma^+$, we write $v^\omega = v v \dots$ for the $\omega$-word arising from $v$ by concatenating it with itself infinitely often.
      An $\omega$-word is \emph{ultimately periodic} if it is of the form $uv^\omega$ for $u \in \Sigma^*$ and $v \in \Sigma^+$.
      Any non-empty $\omega$-regular language contains an ultimately periodic word (see e.\,g.\ \cite{perrin2004infinite}).
      We immediately obtain from \autoref{thm:InfiniteOrbitIsRegular} the following specialization of \cite[Corollary~3.3]{orbitsPart} for automaton semigroups of bounded $S$-activity.
      The more general result \cite[Corollary~3.3]{orbitsPart} states that a general automaton semigroup is infinite if and only if it admits a(n infinite) word with an infinite orbit.
      \begin{corollary}
        The semigroup $\mathscr{S}(\mathcal{T})$ generated by the automaton $\mathcal{T} = (Q, \Sigma, \delta)$ of bounded $S$-activity is infinite if and only if there are $u \in \Sigma^*$ and $v \in \Sigma^+$ with $Q^* \circ uv^\omega$ (i.\,e.\ if it admits an ultimately periodic word with an infinite orbit).
      \end{corollary}
    \end{subsection}
    \begin{subsection}{The Finiteness Problem for Bounded Automaton Semigroups}
      Using \autoref{thm:InfiniteOrbitIsRegular}, we obtain the main result of this paper:
      \begin{theorem}
        The problem
        \problem{
          a complete \SAut $\mathcal{T} = (Q, \Sigma, \delta)$,\newline
          a regular, suffix-closed language $R \subseteq Q^*$ and\newline
          a closed subset $S \subseteq Q$
          such that $S$ generates a finite subsemigroup\footnote{Recall that requiring $S \subseteq Q$ is not a restriction by \autoref{rem:SIsSubsetOfStates}.} and $\mathcal{T}$ is of bounded $S$-activity
        }{
          is the image of $R$ in $\mathscr{S}(\mathcal{T})$ finite?
        }\noindent
        is decidable.
      \end{theorem}
      \begin{proof}
        Since $R$ is suffix-closed, the subset of $\mathscr{S}(\mathcal{T})$ induced by projecting $R$ into the semigroup is infinite if and only if there is an $\omega$-word with an infinite $R$-orbit \cite[Theorem 3.2]{orbitsPart}.\footnote{This is precisely the statement of the referenced theorem.} Thus, the stated decision problem is equivalent to testing whether such an $\omega$-word exists. Since the language of such words is effectively deterministic $\omega$-regular by \autoref{thm:InfiniteOrbitIsRegular}, this problem is decidable (as the emptiness problem for (deterministic) Büchi acceptors is decidable, see e.\,g.\ \cite[Proposition 10.11]{perrin2004infinite}).
      \end{proof}
      
      An immediate consequence of this result is that the finiteness problem for complete automaton semigroups of bounded $S$-activity is decidable.
      \begin{corollary}
        The finiteness problem for \SAuta of bounded activity
        \problem{
          a complete \SAut $\mathcal{T} = (Q, \Sigma, \delta)$ and\newline
          a closed subset $S \subseteq Q$
          such that $S$ generates a finite subsemigroup and $\mathcal{T}$ is of bounded $S$-activity
        }{
          is $\mathscr{S}(\mathcal{T})$ finite?
        }\noindent
        is decidable.
      \end{corollary}
      
      Additionally, we also get that the finiteness problem for finitely generated subsemigroups of complete automaton semigroups of bounded $S$-activity is decidable. This in particular includes the (uniform) order problem for complete automaton semigroups of bounded $S$-activity (compare to \cite{bartholdi2020hierarchy}).
      \begin{corollary}
        The subsemigroup finiteness problem for \SAuta of bounded activity
        \problem{
          a complete \SAut $\mathcal{T} = (Q, \Sigma, \delta)$,\newline
          a finite set $\bm{R} \subseteq Q^*$ and\newline
          a closed subset $S \subseteq Q$
          such that $S$ generates a finite subsemigroup and $\mathcal{T}$ is of bounded $S$-activity
        }{
          is the subsemigroup of $\mathscr{S}(\mathcal{T})$ generated by $\bm{R}$ finite?
          }\noindent
        is decidable.
      \end{corollary}
      
      We also obtain some algorithmic consequences for the semigroup generated by the dual automaton. The \emph{dual} of the complete \SAut $\mathcal{T} = (Q, \Sigma, \delta)$ is the complete \SAut $\partial\mathcal{T} = (\Sigma, Q, \partial\delta)$ with
      \[
        \partial\delta = \{ \trans{a}{p}{q}{b} \mid \trans{p}{a}{b}{q} \in \delta \} \text{.}
      \]
      Concerning the dual, we have the following results regarding the existence of elements with (or without) torsion.
      \begin{corollary}
        The problems
        \problem{
          a complete \SAut $\mathcal{T} = (Q, \Sigma, \delta)$ and\newline
          a closed subset $S \subseteq Q$
          such that $S$ generates a finite subsemigroup and $\mathcal{T}$ is of bounded $S$-activity
        }{
          does $\mathscr{S}(\partial \mathcal{T})$ contain an element without torsion?
        }\noindent
        and
        \problem{
          a complete \SAut $\mathcal{T} = (Q, \Sigma, \delta)$ and\newline
          a closed subset $S \subseteq Q$
          such that $S$ generates a finite subsemigroup and $\mathcal{T}$ is of bounded $S$-activity
        }{
          is $\mathscr{S}(\partial \mathcal{T})$ torsion-free?
        }\noindent
        are decidable.
      \end{corollary}
      \begin{proof}
        \cite[Theorem~3.12]{orbitsPart} yields that $\mathscr{S}(\partial \mathcal{T})$ contains an element of torsion if and only if there is a periodic word $u^\omega$ (with $u \in \Sigma^+$) whose orbit $Q^* \circ u^\omega$ is finite. Similarly, it also yields that $\mathscr{S}(\partial \mathcal{T})$ contains an element without torsion if and only if $Q^* \circ u^\omega$ is infinite for some $u \in \Sigma^+$. Therefore, the two stated problems boil down to testing whether an $\omega$-regular language (or its complement, which is also $\omega$-regular \cite{perrin2004infinite}) contains a periodic word. This, however, is decidable (using standard automaton theory techniques).
      \end{proof}
    \end{subsection}
  \end{section}
  
  \bibliographystyle{plain}
  \bibliography{references}
\end{document}